\documentclass{article}
\usepackage{graphicx} % Required for inserting images
\usepackage{amsthm}
\usepackage{amsmath}
\usepackage{amssymb}
\usepackage{comment}
\usepackage{algorithm}
\usepackage{multirow}
\usepackage{authblk}

\title{Establishing an $\Omega(\sqrt{d})$ complexity lower bound for PDMP samplers and how to break it: a sub-$\sqrt{d}$ algorithm for Gaussian-tailed targets}
\author{Augustin Chevallier}
\affil{Université de Strasbourg \\ a.chevallier@unistra.fr}
\newtheorem{proposition}{Proposition}
\newtheorem{lemma}{Lemma}
\newtheorem{assumption}{Assumption}
\newtheorem{theorem}{Theorem}

\newtheorem{conjecture}{Conjecture}
\newtheorem{remark}{Remark}
\newtheorem{definition}{Definition}

\begin{document}

\maketitle

\begin{abstract}
Despite the theoretical appeal of their non-reversibility, to date, no Piecewise Deterministic Markov Process (PDMP) samplers have been developed that scale better than $\mathcal{O}(\sqrt{d})$ in computational complexity with respect to the target dimension $d$. We prove that this is a fundamental limitation by establishing an $\Omega(\sqrt{d})$ lower bound on the algorithmic complexity of PDMP samplers in a standard setup. By relaxing the assumption that the target density must remain invariant at all continuous times, we then demonstrate how to bypass this barrier. Specifically, we introduce a novel PDMP sampling scheme and show that it achieves an empirical complexity of $\mathcal{O}(d^\alpha)$, where $\alpha \in [0.2, 0.3]$ for Gaussian-tailed targets. In addition, this PDMP scheme is locally adaptive in both trajectory length and distance between velocity updates.
\end{abstract}

\section{Introduction}

Markov Chain Monte Carlo (MCMC) methods are standard tools in Bayesian inference for sampling from complex probability distributions. Piecewise Deterministic Markov Process (PDMP) samplers, such as the Bouncy Particle Sampler and the Zig-Zag process, have emerged as a distinct class of continuous-time, non-reversible samplers, offering an alternative to traditional reversible chains \cite{fearnhead2018piecewise,bouchard2018bouncy,bierkens2019zig}. To improve practical implementation, recent work \cite{chevallier2025practicalpdmpsamplingmetropolis} has expanded this framework in two distinct ways: by adapting the local trajectory-length adaptation of the No-U-Turn Sampler (NUTS) \cite{hoffman2014no} to PDMPs, and by introducing Metropolisation to allow the use of numerical approximations of the event rate. 

However, the practical viability of any sampler in high dimensions relies heavily on its computational scaling with the dimension $d$. While standard reversible algorithms like the Random Walk Metropolis \cite{roberts1997weak} scale as $\mathcal{O}(d)$, gradient-based methods such as the Metropolis-Adjusted Langevin Algorithm (MALA) \cite{roberts1998optimal} and Hamiltonian Monte Carlo (HMC) \cite{neal2011mcmc, betancourt2017conceptual} achieve scalings of $\mathcal{O}(d^{1/3})$ and $\mathcal{O}(d^{1/4})$ \cite{beskos2013optimal}, respectively. In contrast, the Bouncy Particle Sampler scales as $\mathcal{O}(d^{1/2})$ \cite{deligiannidis2021scaling}, and no existing PDMP achieves a better computational complexity.

In this work, we establish an $\Omega(d^{1/2})$ lower bound on the computational complexity for a broad class of local PDMP samplers. This result demonstrates that the $\mathcal{O}(d^{1/2})$ scaling is an inherent limit of the standard PDMP framework rather than an artifact of specific algorithmic designs. Given that other gradient-based methods, such as HMC and MALA, achieve better dimensional scaling, it is natural to wonder why this barrier exists for PDMPs. Our intuition is that this limitation arises because PDMP samplers require the target distribution to remain invariant at all continuous times, a highly stringent constraint.

Therefore, relaxing the continuous-time invariance requirement provides a pathway to surpass the $\mathcal{O}(d^{1/2})$ limit. We demonstrate this by designing a novel PDMP algorithm that forgoes this requirement. Inspired by the numerical approximation of Hamiltonian dynamics, we incorporate this approach into the Metropolised NUTS-PDMP framework \cite{chevallier2025practicalpdmpsamplingmetropolis}.
Empirically, we show that this method achieves a computational complexity of $\mathcal{O}(d^\alpha)$, where $\alpha \in [0.2, 0.3]$ for Gaussian-tailed targets. Beyond this improved scaling, our approach naturally adapts its step size to efficiently explore state spaces with varying geometric characteristics. While standard reversible HMC requires complex strategies to achieve this \cite{bourabee2025gistgibbsselftuninglocally, bou2024incorporating,Modi2024,modi2024atlas, turok2024sampling,bourabee2025withinorbitadaptiveleapfrognouturn}, our framework provides a straightforward mechanism that avoids these complexities.

This paper is organized as follows. Section~\ref{sec:previous_work} reviews the necessary background on Piecewise Deterministic Markov Processes and the Metropolised NUTS PDMP framework. Section~\ref{sec:complexity_bound} establishes the theoretical complexity lower bound for standard PDMP samplers. Section~\ref{sec:algorithm} introduces our novel sub-$\sqrt{d}$ algorithm, followed by Section~\ref{sec:empirical_results} which presents empirical scaling results. Finally, Section~\ref{sec:discussion} concludes with a discussion of our findings.
\section{Previous work}
\label{sec:previous_work}

This section introduces Piecewise Deterministic Markov Processes (PDMPs) for sampling and outlines the Metropolised NUTS-PDMP framework. Readers interested solely in the complexity results of Section \ref{sec:complexity_bound} may focus entirely on Section \ref{sec:pdmp_for_sampling} and skip Sections \ref{sec:metropolised_pdmp} and \ref{sec:nuts_pdmp}.

\subsection{PDMPs for MCMC Sampling}
\label{sec:pdmp_for_sampling}

In Markov Chain Monte Carlo (MCMC) sampling, the goal is to draw samples from a target measure $\pi(x)$ defined on $\mathbb{R}^d$. Much like traditional discrete-time MCMC methods, Piecewise Deterministic Markov Processes (PDMPs) \cite{fearnhead2018piecewise} can be used to construct continuous-time Markov processes such that their invariant distribution is the target $\pi$.

A PDMP evolves exactly as its name implies: it moves along deterministic trajectories until a random event occurs, at which point the process jumps to a new state. Concretely, these dynamics are governed by three corresponding components:
\begin{enumerate}
    \item[i)] the driving deterministic motion, described by a vector field $\phi$;
    \item[ii)] the event rate $\lambda$, which designates the state-dependent intensity of a jump occurring;
    \item[iii)] the jump kernel $Q$, which specifies the distribution of the new state following an event.
\end{enumerate}

To ensure the process targets the desired invariant measure $\mu$, the choices of $\phi$, $\lambda$, and $Q$ must satisfy the stationary equation. This relationship is formalized using the adjoint of the PDMP generator \cite{fearnhead2018piecewise}. For a PDMP on state space $E$, the adjoint $L^*$ applied to a probability density $\mu$ is:
\begin{equation}
    \label{eq:adjoint-generator}
    L^*\mu(x) = -\text{div}(\phi(x)\mu(x)) + \int_E \lambda(y) \mu(y) Q(y, dx) - \lambda(x)\mu(x).
\end{equation}
The measure $\mu$ is invariant if and only if $L^*\mu(x) = 0$ for all $x \in E$. This equation is central to establishing the complexity bounds in Section \ref{sec:complexity_bound}.

To construct a PDMP that successfully explores the space and targets $\pi$, existing samplers typically augment the state space with an auxiliary velocity variable $v \in \mathbb{R}^d$. The process runs on the extended phase space $E = \mathbb{R}^d \times \mathbb{R}^d$, targeting the joint invariant measure $\mu(x, v) = \pi(x)\rho(v)$, where $\rho(v)$ is a chosen velocity distribution. 
While other variants exist, the two most well known PDMP samplers are the Bouncy Particle Sampler (BPS) \cite{bouchard2018bouncy} and the Zig-Zag Sampler (ZZS) \cite{bierkens2019zig}. Both rely on the auxiliary velocity space, and both employ constant-velocity dynamics giving the vector field $\phi(x,v) = (v, 0)$.

\paragraph{Bouncy Particle Sampler (BPS):} The velocity distribution $\rho(v)$ is typically uniform on the unit sphere. The reflection rate is $\lambda(x,v) = \max(0, v \cdot \nabla U(x))$. When an event occurs, the jump kernel $Q$ reflects the velocity across the hyperplane orthogonal to $\nabla U(x)$ using the deterministic flip operator $F(x, v) = \left(x, v - 2\frac{v \cdot \nabla U(x)}{\|\nabla U(x)\|^2}\nabla U(x)\right)$. To maintain ergodicity, BPS also resamples $v$ from $\rho(v)$ at a constant rate.
    
\paragraph{Zig-Zag Sampler (ZZS):} The velocity space is restricted to $\{-1, +1\}^d$. ZZS flips each coordinate $v_i$ independently with rate $\lambda_i(x, v) = \max(0, v_i \partial_{x_i} U(x))$. The kernel $Q_i$ applies a deterministic flip operator $F_i$ that reverses the sign of $v_i$ ($v_i \mapsto -v_i$) while leaving the position and other velocity components unchanged.

\subsection{Metropolised PDMP}
\label{sec:metropolised_pdmp}

In a PDMP, the event rate varies continuously along the deterministic trajectory. Because computing the time of the next jump is typically intractable, traditional simulation relies on Poisson thinning. This technique proposes jump times using an upper bound of the rate, which generates extra events. These proposed events are then probabilistically rejected, or \textit{thinned}, to match the true event distribution. However, this approach requires finding an upper bound for the rate, either analytically or numerically \cite{sutton2023concave,corbella2022automatic,andral2024automated}, and is a major difficulty for the pratical use of PDMP samplers. To bypass this difficulty, the algorithm presented in section \ref{sec:algorithm} relies on the Metropolised PDMP framework \cite{chevallier2025practicalpdmpsamplingmetropolis}, a different approach that does not rely on Poisson thinning and eliminates the need for upper bounds.
To ease the exposition, we introduce this framework in a simplified scenario characterized by:
\begin{enumerate}
    \item an invariant target distribution with density $\mu$,
    \item a continuous deterministic flow that preserves the Lebesgue measure,
    \item a single event rate function $\lambda$,
    \item a jump kernel $Q(\cdot|z) = \delta_{F(z)}$ where $F$ a volume-preserving mapping $F$ which satisfies $\mu(F(z)) = \mu(z)$ for any state $z$. In other words, whenever an event is occurs, the process undergoes a deterministic transition from state $z$ to $F(z)$.
\end{enumerate}

The core idea of the Metropolised PDMP approach is to use a numerically approximated PDMP path to build a proposal within a Metropolis-Hastings scheme. It is assumed that the continuous deterministic dynamics can be solved analytically, but the true rate $\lambda$ is substituted by a tractable numerical approximation $\tilde{\lambda}$ along the trajectory (for example, using a piecewise-constant scheme). For a given starting point $z_0$, the proposal is therefore $\tilde{z}_T$ the value at a fixed time $T$ of the simulated PDMP starting at $z_0$.

\paragraph{Path space}
In order to formulate the Metropolis-Hastings acceptance ratio, we must define the probability measure governing proposed paths. This requires a formal specification of the path space for times in $[0,T]$ for a fixed $T$. Since the jump kernel is deterministic, a generated trajectory is entirely determined by its initial position $z_0$ and the sequence of jump times $0 \leq t_1 < \dots < t_k \leq T$. Consequently, the path space is defined as:
\[
    W = E \times \bigcup_{k \in \mathbb{N}} \left(\, ]0,T[ \,\right)^k,
\]
where $E$ represents the state space of the process. Equipped with the Lebesgue measure on $]0,T[^k$ for each partition, we can express probability densities on $W$. Given a starting state $z_0$, the density of a path $\tilde{w} = (z_0, t_1, \dots, t_k)$ is given by:
\begin{equation}
    \label{eq:proba-path}
    \tilde{p}(\tilde{w} \mid z_0)
    = \left[ \prod_{i=1}^k 
        \tilde{\lambda}(\tilde{z}_{t_i^-}) 
        \exp\!\left( -\int_{t_{i-1}}^{t_i} \tilde{\lambda}(\tilde{z}_s)\,ds \right)  \right]
        \exp\!\left( -\int_{t_k}^{T} \tilde{\lambda}(\tilde{z}_s)\,ds \right).
\end{equation}

\paragraph{Backward process}
To compute the acceptance ratio, the framework do not integrate over all possible path going from some $z_0$ to some $\tilde{z}_T$, as this would be intractible. It however relies on the probablity of using the exact same path but backward. Of course, since PDMP are non reversible, this is usualy 0. To solve this issue, the framework relies on using a numerical approximation of the time reversal of the process. By leveraging the time-reversal properties of PDMPs \cite{lopker2013}, the backward process can be similarly defined on $W$. Under our simplified assumptions, the reversed PDMP process is characterized by:
\begin{enumerate}
    \item following the deterministic flow in reverse time,
    \item an event rate defined by $\lambda^r(z) = \lambda(F^{-1}(z))$,
    \item and a jump kernel given by $Q^r(\cdot|z) = \delta_{F^{-1}(z)}$.%$z \mapsto F^{-1}(z)$.
\end{enumerate}
The path-reversal mapping $R: W \to W$ is given by:
\[
    R: (z_0, t_1, \dots, t_k) \mapsto (\tilde{z}_T, T - t_k, \dots, T - t_1).
\]
\begin{remark}
    In our simplified setting, the Jacobian determinant (volume change) associated with the mapping $R$ evaluates strictly to $1$, since both the continuous deterministic flow and the jump transition kernel $F$ are assumed to preserve the Lebesgue measure. However, when adapting this framework to non-volume-preserving flows or alternative jump kernels, one must be careful to compute and include the appropriate Jacobian determinant into the Metropolis--Hastings acceptance ratio.
\end{remark}

\paragraph{Algorithm}
Starting from an initial state $z_0 \in E$, a single step of the Metropolised PDMP algorithm proceeds as follows. First, pick a time direction $\gamma \in \{-1, +1\}$ uniformly at random, corresponding to forward ($+1$) or backward ($-1$) simulation. Second, depending on $\gamma$, generate a path starting from $z_0$ using either the approximated forward process ($\tilde{w} = (\tilde{z}_t)_{t \in [0,T]}$) or backward process ($\tilde{w}^r = (\tilde{z}_t^r)_{t \in [0,T]}$). Finally, the terminal state of the generated path ($\tilde{z}_T$ or $\tilde{z}_T^r$) is accepted with probability:
\begin{equation}
    \label{eq:acceptance-proba}
    \alpha(\tilde{w}) =
    \begin{cases}
    1 \wedge \dfrac{\mu(\tilde{z}_T)\,\tilde{p}^r(R(\tilde{w})\mid \tilde{w}_T)}
    {\mu(z_0)\,\tilde{p}(\tilde{w} \mid z_0)}, & \gamma = 1, \\[1.5em]
    1 \wedge \dfrac{\mu(\tilde{z}_T^r)\,\tilde{p}(R(\tilde{w}^r)\mid \tilde{w}_T^r)}
    {\mu(z_0)\,\tilde{p}^r(\tilde{w}^r \mid z_0)}, & \gamma = -1,
    \end{cases}
\end{equation}
where $\tilde{p}^r$ denotes the path-space probability density for the approximated reverse process. This expression uses the fact that the volume change associated with the reversal map $R$ is exactly $1$.

Crucially, there are no restrictions on how the numerical approximation is constructed, provided it does not depend on past steps. Therefore, it is possible to design an approximation of the rate where the step size is chosen based on the local properties of the target.

\subsection{No-U-Turn PDMP}
\label{sec:nuts_pdmp}

While the Metropolised PDMP framework provides a rigorous way to simulate with approximate rates, the trajectory duration $T$ is fixed, while the optimal path length might varies significantly across different regions of the state space. The No-U-Turn PDMP (NUTS-PDMP) \cite{chevallier2025practicalpdmpsamplingmetropolis} addresses this by adapting the No-U-Turn methodology developped for HMC to PDMP samplers, providing a locally adaptive algorithm to chose $T$.
First, the No-U-Turn criterion is adapted to PDMPs in the following definition.

\begin{definition}[PDMP No-U-Turn criterion]
    Let $Z = (x_t, v_t)_{t \in \mathbb{R}}$ be a PDMP path experiencing velocity jumps at event times $\mathcal{E} \subset \mathbb{R}$. The path satisfies the PDMP No-U-Turn Criterion on an interval $[a, b]$ if for all event times $t_i, t_j \in \mathcal{E} \cap [a, b]$, the following conditions hold:
    \begin{align*}
        (x_{t_i} - x_{t_j}) \cdot v_{t_j^-} &< 0 \quad \text{for all } t_i < t_j, \\
        (x_{t_i} - x_{t_j}) \cdot v_{t_j^+} &< 0 \quad \text{for all } t_i < t_j \text{ with } t_j < b, \\
        (x_{t_i} - x_{t_j}) \cdot v_{t_j^+} &> 0 \quad \text{for all } t_i > t_j, \\
        (x_{t_i} - x_{t_j}) \cdot v_{t_j^-} &> 0 \quad \text{for all } t_i > t_j \text{ with } t_i > a.
    \end{align*}
\end{definition}

A single step of PDMP-NUTS is decomposed into two substeps. The first step generates a trajectory until the NUTS criterion is no longer valid. The second step resamples a time $t^*$ on this trajectory using a rule that ensures the law of the state at time $t^*$ is the target distribution.

\paragraph{Trajectory generation}
Starting from an initial state $z_0$:
\begin{enumerate}
    \item Sample a fraction $\alpha$ uniformly in $[0,1]$.
    \item Find the largest $T$ such that simulating the exact PDMP forward in time for a duration of $(1-\alpha)T$ and backward in time for $\alpha T$ satisfies the No-U-Turn criterion on $[-\alpha T, (1-\alpha)T]$. This generates a trajectory $\bar{Y}$ parameterized on $[-\alpha T, (1-\alpha)T]$ with $\bar{Y}_0 = z_0$. 
    \item We then define a time-shifted trajectory $Y$ on $[0, T]$ given by $Y_t = \bar{Y}_{t - \alpha T}$, with $l = \alpha T$ being the shift.
\end{enumerate}

\paragraph{Resampling step}
From now on, we assume that the law of $z_0$ is the target $\mu$. Since $l$ is the starting time on the shifted trajectory, $Y_l \sim \mu$. Therefore, one can ask: knowing $Y$ and that the law of the starting point is $\mu$, what are the possible starting times that could generate $Y$? This is equivalent to determining the conditional density $w$ of the shift $l$ given the trajectory $Y$. Then, any $t^*$ sampled from $w$ will satisfy that the law of $Y_{t^*}$ is the same as the law of $Y_l$, which means it is $\mu$. 

It is established in \cite{chevallier2025practicalpdmpsamplingmetropolis} that
\begin{equation}
    w(t) \propto \mu(Y_t) \, p_{W_{[t,T]}}(Y \mid Y_t) \, p^r_{W_{[0,t]}}(R(Y_{[0,t]}) \mid Y_t) \, \nu(t),
\end{equation}
where $p_{W_{[t,T]}}(\cdot \mid Y_t)$ is the probability density of generating the forward segment using the standard forward process starting from $Y_t$, and $p^r_{W_{[0,t]}}(\cdot \mid Y_t)$ is the probability density of generating the backward segment using the reverse process starting from $Y_t$. Here, $R$ denotes the time-reversal operator mapping a forward-parameterized segment to its time-reversed counterpart. The term $\nu(t)$ is a change-of-volume factor induced by the discrete stopping criterion: specifically, $\nu(t) = \frac{T-t}{T^2}$ if the criterion was violated in the forward direction, and $\nu(t) = \frac{t}{T^2}$ if violated in the backward direction.

When the PDMP is simulated exactly (and admits $\mu$ as its invariant distribution), the path densities cancel with the target measure, and this expression simplifies to:
\begin{equation}
    \label{eq:exact_w}
    w(t) \propto \nu(t).
\end{equation}
However, when the PDMP dynamics are approximated, we must evaluate the full unnormalized density:
\begin{equation}
    w(t) \propto \mu(Y_t) \, \tilde{p}_{W_{[t,T]}}(Y \mid Y_t) \, \tilde{p}^r_{W_{[0,t]}}(R(Y_{[0,t]}) \mid Y_t) \, \nu(t).
\end{equation}
This density can be evaluated pointwise using the path probabilities defined in the previous section. To sample from $w(t)$, we can employ a Metropolis-Hastings step. Given that $w(t)$ is exactly proportional to $\nu(t)$ in the exact setting (Equation \ref{eq:exact_w}), $\nu(t)$ serves as a natural proposal distribution, as we expect the numerical approximation not to drastically modify $w(t)$.

Hence, the complete resampling procedure proceeds as follows:
\begin{enumerate}
    \item Define the unnormalized conditional density function $w(t)$ on $[0,T]$, representing the probability of the starting time being $t$ given the generated trajectory $Y$.
    \item Compute the unnormalized weight at the current time shift, $w(l)$.
    \item Sample a proposed time $t_{prop}$ from the proposal density $\nu(t)$, and compute its corresponding weight $w(t_{prop})$.
    \item Accept the proposed time with the Metropolis-Hastings probability:
    \begin{equation}
        \alpha = 1 \wedge \frac{w(t_{prop})\nu(l)}{w(l)\nu(t_{prop})}.
    \end{equation}
    Let $t^* = t_{prop}$ if the proposal is accepted, and $t^* = l$ if it is rejected.
    \item Return $Y_{t^*}$ as the next state in the Markov chain.
\end{enumerate}

\section{Complexity lower bound for PDMP sampling algorithms}
\label{sec:complexity_bound}
The primary goal of this section is to establish an asymptotic lower bound on the computational complexity of Piecewise Deterministic Markov Process (PDMP) sampling algorithms as the spatial dimension $d \rightarrow \infty$. Evaluating this limit requires specifying a family of target distributions indexed by $d$, as the resulting complexity lower bound depends heavily on this choice. In this analysis, we focus on i.i.d. targets, assuming $\pi$ is the product measure of a one-dimensional marginal distribution $\pi_1$ such that $\pi(x) = \prod_{i=1}^d \pi_1(x_i)$. Before deriving this lower bound, we must specify what constitutes a PDMP sampling algorithm and how its computational cost is measured.

In what follows, we consider a target measure $\pi$ on $\mathbb{R}^d$ and a PDMP defined on an extended state space $E = \mathbb{R}^d \times \mathcal{Z}$, where $\mathcal{Z}$ is an auxiliary space. We denote by $\mu$ a joint distribution on $E$ such that $\pi$ is its marginal over $\mathbb{R}^d$, and by $\rho(\cdot \mid x)$ the corresponding conditional distribution on $\mathcal{Z}$. Let $\phi = (\phi_x,\phi_z)$ be the deterministic vector field governing the continuous dynamics on $\mathbb{R}^d \times \mathcal{Z}$, $\lambda$ be the jump rate, and $Q$ be the jump kernel.

With these components established, a PDMP sampling algorithm can be broadly defined as a mapping that constructs a process tailored to the target distribution:
\[
    \pi \mapsto (\mathcal{Z},\phi,\rho, \lambda, Q).
\]

Before analyzing specific algorithms, we must define their operational objective. The fundamental goal of any MCMC algorithm is to generate a sequence of independent and identically distributed (i.i.d.) samples from the target measure $\pi$. While a continuous-time Markov process generates inherently correlated spatial trajectories, practitioners measure efficiency by extracting samples spaced sufficiently far apart in time to be considered \textit{effectively independent}. 

\begin{assumption}[Effectively Independent Sampling Times]
    \label{assum:independent_times}
    There exists an infinite sequence of random times $0 \leq t_1 < t_2 < \dots < t_k < \dots$ increasing almost surely to infinity, such that the sequence of spatial states $(X_{t_k})_{k \geq 1}$ acts as a sequence of independent draws from $\pi$. Specifically, we assume that the empirical average of the distances between consecutive points converges almost surely to the expected distance between two independent random vectors distributed according to $\pi$:
    \[
        \lim_{K \to \infty} \frac{1}{K} \sum_{k=1}^{K-1} \|X_{t_{k+1}} - X_{t_k}\| = \mathbb{E}_{\pi \otimes \pi}[\|X - Y\|] \quad a.s.
    \]
\end{assumption}

The complexity of a PDMP sampling algorithm is then defined as the average computation required to generate one effectively independent sample. In the context of PDMPs, this computational cost is the number of jump events. In practice, the number of events serves as a strict lower bound on the cost: simulating event times (whether via thinning or Metropolised approaches) requires gradient evaluations at least at every event.

While the mapping $\pi \mapsto (\mathcal{Z},\phi,\rho, \lambda, Q)$ is highly general, we restrict our complexity analysis to a subset of reasonable algorithms that satisfy the following structural assumptions.

\begin{assumption}[Target-independence of the flow and auxiliary space]
\begin{enumerate}
    \item The auxiliary space $\mathcal{Z}$ may depend only on the dimension $d$ of the state space. 
    \item The deterministic vector field $\phi$ and the auxiliary space $\mathcal{Z}$ are independent of the target distribution $\pi$. 
    (Note: This ensures the algorithm relies on general-purpose, tractable flows rather than target-specific flows, such as exact Hamiltonian trajectories, which may lack closed-form solutions).
\end{enumerate}
\end{assumption}

\begin{comment}
\begin{assumption}[Locality of the auxiliary distribution]
    The conditional probability distribution $\rho(\cdot \mid x)$ depends on $x$ exclusively through local evaluations of the target $\pi$ and its derivatives at $x$ (e.g., $\pi(x), \nabla \pi(x)$), rather than through an explicit spatial dependence on $x$ or global properties of $\pi$.
\end{assumption}
\end{comment}

\begin{assumption}[Local jumps]
    The jump kernel $Q$ can only change the auxiliary variable $z$, and not the spatial position $x$. Consequently, the spatial trajectories $t \mapsto X_t$ are absolutely continuous, as is standard for PDMP-based sampling algorithms.
\end{assumption}

\subsection{Preliminary results}

\begin{assumption}[Ergodicity]
    \label{assum:ergodicity}
    The PDMP $(X_t, Z_t)_{t \geq 0}$ is ergodic with respect to its invariant measure $\mu$. Specifically, for any $\mu$-integrable function $f$, the time average along the trajectory converges almost surely to the spatial expectation:
    \[
        \lim_{S \to \infty} \frac{1}{S} \int_0^S f(X_t, Z_t) dt = \mathbb{E}_\mu[f] \quad a.s.
    \]
\end{assumption}

\begin{proposition}[Lower Bound on PDMP Sampling Complexity for i.i.d. targets]
    \label{prop:complexity_lower_bound}
    Let the target distribution $\pi$ on $\mathbb{R}^d$ be an i.i.d. product measure such that $\pi(x) = \prod_{i=1}^d \pi_1(x_i)$. Let the PDMP satisfy Assumption \ref{assum:ergodicity} and Assumption \ref{assum:independent_times}. 
    
    The sampling complexity $C_d$, defined as the almost sure asymptotic number of jump events required per effectively independent sample ($C_d = \lim_{K \to \infty} \frac{N(t_K)}{K}$), satisfies the following lower bound:
    \[
        C_d \geq \frac{C(\pi_1)\sqrt{d}}{\mathbb{E}_\mu[\|\phi_x\|]} \mathbb{E}_\mu[\lambda]
    \]
    where $C(\pi_1)$ is the constant defined in Lemma \ref{lem:distance_lower}.
\end{proposition}
\begin{proof}[Proof Sketch]
    The full proof is deferred to Appendix \ref{app:proofs}, and only the general ideas are given here.
    By ergodicity, the average number of jumps per unit of time along a trajectory is given by $\mathbb{E}_\mu[\lambda]$.
    
    The complexity $C_d$ is the average number of jumps per unit of time, multiplied by $\tau$ the average time between two effectively independent samples:
    \[
        C_d = \mathbb{E}_\mu[\lambda] \tau
    \]
    The distance between two independant samples grows as $\sqrt{d}$, therefore the travelled distanced between two independant samples is at least $C(\pi_1)\sqrt{d}$ for some constant $C(\pi_1)$ only depending on $\pi_1$. Furtermore the average speed is $\mathbb{E}_\mu[\|\phi_x\|]$, hence $\tau \geq \frac{C(\pi_1)\sqrt{d}}{\mathbb{E}_\mu[\|\phi_x\|]}$, which concludes the proof.
\end{proof}

\begin{lemma}[0 mean divergence]
    \label{lemma:zero-mean-div}
    Let $\mu$ be a probability distribution on some vector space $E$, and $x\mapsto \phi(x), x \mapsto \lambda(x), c\mapsto Q(x,\cdot)$ be a vector field, a jump rate and a jump kernel of a PDMP. If $\mu$ is invariant by the PDMP then:
    \[
        E_\mu[\frac{div(\phi \mu)}{\mu}] = 0
    \]
\end{lemma}
\begin{proof}
    The proof is deferred to Appendix \ref{app:proofs}.
\end{proof}

\begin{proposition}[lower bound of $\lambda$]
    \label{prop:lower-bound}
    Let $\mu$ be a probability distribution on some vector space $E$, and $x\mapsto \phi(x), x \mapsto \lambda(x), c\mapsto Q(x,\cdot)$ be a vector field, a jump rate and a jump kernel of a PDMP. If $\mu$ is invariant by the PDMP then for all $x$:
    \[
        \lambda(x) \geq max(0, -div(\phi \mu)(x) / \mu(x)). 
    \]
    This implies
    \[
        E_\mu[\lambda(x)] \geq \frac{1}{2} E_\mu[\frac{|div(\phi \mu)|}{\mu}].
    \]
\end{proposition}
\begin{proof}
    From the stationary equation (Equation \ref{eq:adjoint-generator}), the adjoint of the generator applied to the invariant measure $\mu$ evaluates to zero.
    Rearranging to isolate the outgoing jump term yields:
    \[
        \lambda(x)\mu(x) = -\text{div}(\phi(x)\mu(x)) + \int_E \lambda(y) \mu(y) Q(y, dx).
    \]
    Since the jump rate $\lambda(y)$, measure $\mu(y)$, and transition kernel $Q(y, dx)$ are all non-negative, the integral representing incoming jumps is non-negative ($\geq 0$). Dropping this term provides a lower bound:
    \[
        \lambda(x)\mu(x) \geq -\text{div}(\phi(x)\mu(x)) \implies \lambda(x) \geq -\frac{\text{div}(\phi(x)\mu(x))}{\mu(x)}.
    \]
    Because the jump rate itself is strictly non-negative ($\lambda(x) \geq 0$), we combine the bounds:
    \[
        \lambda(x) \geq \max\left(0, -\frac{\text{div}(\phi \mu)(x)}{\mu(x)}\right).
    \]
    
    For the expectation inequality, define $g(x) = \frac{\text{div}(\phi(x)\mu(x))}{\mu(x)}$. From Lemma \ref{lemma:zero-mean-div}, $E_\mu[g(x)] = 0$. 
    Decomposing $g(x)$ into its positive and negative parts, $g(x) = g^+(x) - g^-(x)$, the zero-mean property implies the parts are balanced: $E_\mu[g^+] = E_\mu[g^-]$. Consequently, the expectation of the negative part is exactly half the expectation of the absolute value: 
    \[
        E_\mu[g^-] = \frac{1}{2} E_\mu[|g|].
    \]
    Using our established bound, $\lambda(x) \geq \max(0, -g(x)) = g^-(x)$. Taking the expectation with respect to $\mu$ yields:
    \[
        E_\mu[\lambda(x)] \geq E_\mu[g^-(x)] = \frac{1}{2} E_\mu[|g(x)|] = \frac{1}{2} E_\mu\left[\frac{|\text{div}(\phi \mu)|}{\mu}\right].
    \]
    This completes the proof.
\end{proof}

\subsection{Main theorem}

We restrict ourselves to a smaller, but practically relevant, class of PDMPs that cover all existing algorithms. Specifically, we assume that the sampling algorithm exploits the independence structure of the target distribution as follows:

\begin{assumption}[Independence Structure for i.i.d. Targets]
    \label{assum:independence_structure}
    First, the state space for a given dimension $d$ can be written as $(\mathbb{R} \times \mathcal{Z}_1)^d$, where $\mathcal{Z}_1$ is the auxiliary space used for targets in $\mathbb{R}$. 
    
    Second, for an i.i.d. target distribution---meaning $\pi(x) = \prod_{i=1}^d \pi_1(x_i)$---the algorithm yields a structure on each marginal space $\mathbb{R} \times \mathcal{Z}_1$ that is identical to what would be used in the 1D case. This implies:
    \begin{enumerate}
        \item The auxiliary conditional distribution completely factorizes: $\rho(x,z) = \prod_{i=1}^d \rho_1(x_i,z_i)$, leading to a joint invariant measure $\mu(x,z) = \prod_{i=1}^d \mu_1(x_i,z_i)$ with $\mu_1(x_1,z_1) = \pi_1(x_1) \rho_1(x_1,z_1)$.
        \item The deterministic vector field decomposes cleanly: $\phi(x,z) = (\phi_1(x_1,z_1), \dots, \phi_1(x_d,z_d))$, where $\phi_1$ is the exact vector field used in the 1D case.
    \end{enumerate}
\end{assumption}

\begin{theorem}[Complexity for i.i.d. targets]
    \label{th:complexity}
    Under Assumption \ref{assum:independence_structure}, if the divergence of the 1D process has a non-zero finite variance, the expected jump rate scales as $\Omega(\sqrt{d})$, and the overall sampling complexity scales as $C = \Omega(\sqrt{d})$.
\end{theorem}

\begin{proof}
    Let $S_d(x,z) = \frac{\text{div}(\phi \mu)(x,z)}{\mu(x,z)}$. Given the factorized structure of both the vector field $\phi$ and the invariant measure $\mu$, the divergence operator distributes over the individual coordinates:
    $$ \text{div}(\phi\mu) = \sum_{i=1}^d \text{div}_1(\phi_1 \mu_1)(x_i, z_i) \prod_{j \neq i} \mu_1(x_j, z_j) $$

    Dividing by the joint measure $\mu = \prod_{j=1}^d \mu_1(x_j, z_j)$ cancels out the product terms, yielding a sum of independent ratios:
    $$ S_d(x,z) = \sum_{i=1}^d \frac{\text{div}_1(\phi_1 \mu_1)(x_i, z_i)}{\mu_1(x_i, z_i)} $$

    Let $Y_i = \frac{\text{div}_1(\phi_1 \mu_1)(x_i, z_i)}{\mu_1(x_i, z_i)}$. Under the stationary measure $\mu$, the random variables $Y_1, \dots, Y_d$ are independent and identically distributed. By Lemma \ref{lemma:zero-mean-div} applied to the 1D marginals, we know that $E_{\mu_1}[Y_i] = 0$. Let $\sigma^2 = E_{\mu_1}[Y_i^2]$ be the strictly positive, finite variance of these terms.

    By the Central Limit Theorem (CLT), as the dimension $d \to \infty$, the normalized sum $\frac{S_d}{\sigma \sqrt{d}}$ converges in distribution to a standard normal $\mathcal{N}(0, 1)$. To deduce the asymptotic behavior of the expectation $E_\mu[|S_d|]$, we must establish that the sequence $\left\{ \frac{S_d}{\sigma \sqrt{d}} \right\}_{d \ge 1}$ is uniformly integrable. 

    Because the variables $Y_i$ are independent with mean zero and variance $\sigma^2$, the variance of the normalized sum is constant for all $d$:
    $$ E_\mu\left[ \left( \frac{S_d}{\sigma \sqrt{d}} \right)^2 \right] = \frac{1}{\sigma^2 d} (d \sigma^2) = 1 $$

    Since the sequence $\left\{ \frac{S_d}{\sigma \sqrt{d}} \right\}_{d \ge 1}$ is bounded in $L^2(\mu)$, it is uniformly integrable in $L^1(\mu)$. This permits us to pass the limit inside the expectation, yielding convergence in $L^1$:
    $$ \lim_{d \to \infty} E_\mu\left[ \left| \frac{S_d}{\sigma \sqrt{d}} \right| \right] = E[|\mathcal{N}(0, 1)|] = \sqrt{\frac{2}{\pi}} $$

    Multiplying by the normalization factor $\sigma \sqrt{d}$, we obtain the exact asymptotic equivalence for the expected absolute sum:
    $$ E_\mu[|S_d|] \sim \sigma \sqrt{\frac{2}{\pi}} \sqrt{d} = \Theta(\sqrt{d}) $$

    Using the lower bound established proposition \ref{prop:lower-bound}, the expected jump rate must satisfy:
    $$ E_\mu[\lambda] \geq \frac{1}{2} E_\mu\left[ \frac{|\text{div}(\phi \mu)|}{\mu} \right] = \frac{1}{2} E_\mu[|S_d|] = \Omega(\sqrt{d}) $$

        To evaluate the overall sampling complexity $C_d$, we now invoke Proposition \ref{prop:complexity_lower_bound}, which establishes that:
    \[
        C_d \geq \frac{C \sqrt{d}}{\mathbb{E}_\mu[\|\phi_x\|]} \mathbb{E}_\mu[\lambda]
    \]
    where $C > 0$ is the distance constant from Lemma \ref{lem:distance_lower}. 
    
    We must bound the expected speed $\mathbb{E}_\mu[\|\phi_x\|]$ in the denominator. By our structural independence assumptions, the continuous flow decomposes as $\phi_x(x,z) = (\phi_{x,1}(x_1,z_1), \dots, \phi_{x,1}(x_d,z_d))$. Applying Jensen's inequality, we bound the $L^1$ norm by the $L^2$ norm:
    \[
        \mathbb{E}_\mu[\|\phi_x\|] \le \sqrt{\mathbb{E}_\mu[\|\phi_x\|^2]}  = \sqrt{d \mathbb{E}_{\mu_1}[\phi_{x,1}^2]}
    \]
    Assuming the 1D flow has a finite second moment (let $v_1^2 = \mathbb{E}_{\mu_1}[\phi_{x,1}^2] < \infty$), we have $\mathbb{E}_\mu[\|\phi_x\|] \le v_1 \sqrt{d}$.
    
    Substituting this upper bound into our complexity inequality gives:
    \[
        C_d \geq \frac{C \sqrt{d}}{v_1 \sqrt{d}} \mathbb{E}_\mu[\lambda] = \frac{C}{v_1} \mathbb{E}_\mu[\lambda]
    \]
    Since $C$ and $v_1$ are constants strictly independent of the dimension $d$, and we have already shown that $\mathbb{E}_\mu[\lambda] = \Omega(\sqrt{d})$, it immediately follows that:
    \[
        C_d = \Omega(\sqrt{d})
    \]
    This completes the proof.
\end{proof}

\begin{remark}[BPS]
    While the standard formulation of BPS (using velocities on the unit sphere) does not satisfy the assumptions of Theorem \ref{th:complexity}, BPS can equivalently be defined with normally distributed velocities. Therefore, our theorem still applies.
\end{remark}

\subsection{Conjecture}

While Theorem \ref{th:complexity} establishes the $\Omega(\sqrt{d})$ bound under strict structural independence assumptions, we strongly suspect this limitation is universal for PDMP samplers where the rates and velocity distribution can only exploit local information on the target measure. Fundamentally, for an algorithm to break this complexity barrier on an i.i.d. target, its continuous flow would need to preserve the target measure increasingly well as the dimension grows, which seems implausible. We formally state this intuition in the following conjecture:

\begin{conjecture}
    For any PDMP sampling algorithm satisfying Assumption 1 (Target-independence of the flow and auxiliary space), Assumption 2 (Locality of the auxiliary distribution), and Assumption 3 (Local jumps), there exists a one-dimensional distribution $\pi_1$ (satisfying reasonable assumptions) such that the sampling complexity for the i.i.d. target distribution $\pi(x) = \prod_{i=1}^d \pi_1(x_i)$ on $\mathbb{R}^d$ is bounded below by $\Omega(\sqrt{d})$. 
\end{conjecture}

As a brief remark, if the algorithm is assumed to be permutation-invariant when applied to an i.i.d. target, the variables $Y_i$ introduced in the proof act as exchangeable random variables. This may offer a viable pathway toward proving the conjecture.

\subsection{Breaking the limit}

Although PDMP samplers hold great promise as non-reversible processes, theorem \ref{th:complexity} highlights that traditional PDMPs complexity scales at best as $\mathcal{O}(d^{1/2})$, making them less efficient than HMC ($\mathcal{O}(d^{1/4})$) and Langevin-based samplers ($\mathcal{O}(d^{1/3})$).
To break this limit, we must therefore look beyond the usual setting of PDMP samplers. We provide here three possible directions to do so. The last of these directions is developed in the next section into a novel algorithm, which empirically achieves a scaling between $\mathcal{O}(d^{0.2})$ and $\mathcal{O}(d^{0.3})$ for Gaussian-tailed targets.

\paragraph{Target-dependent dynamics: } The Boomerang sampler is such an algorithm \cite{bierkens2020boomerang}. By introducing a priori knowledge about the target distribution, we can construct dynamics that are more adapted to the target distribution, leading to better sampling complexity. The current limitation of the Boomerang sampler is that it only works for distributions close to gaussian targets, and require manual inputs from the user. Ideally an algorithm should automatically build the mapping $\pi \mapsto \phi$, that is, to choose the dynamics automatically given a target distribution. The author believes it might be possible to use recent advances in LLMs to do so.

\paragraph{Numerical optimisation through Metropolisation: } Following the work of \cite{chevallier2025practicalpdmpsamplingmetropolis}, the rate of the process can be simulated approximately. Therefore, even if the number of events is $O(\sqrt{d})$, it may be possible to build an approximation that remains valid across multiple events using only one evaluation of the gradient. This can be done for ZigZag, for example, reducing the overall complexity from $O(d)$ to something closer to $O(\sqrt{d})$. However, this requires the jump kernel $Q$ to be independent of $\nabla \log \pi$. This is the case for ZigZag, but not for any of the known $O(\sqrt{d})$ PDMP samplers.

\paragraph{Relaxing the invariance: } Current PDMP samplers require the target distribution to be invariant for all $t \in \mathbb{R}$. This is a much stronger constraint than for traditional MCMC, where the invariance is only required at discrete times. For the author, this is the fundamental reason why PDMP samplers equipped with the natural structure described above scale as $O(\sqrt{d})$, while MCMC samplers using a similar setting (that is, a product form for $\rho$, such as HMC) have better scaling. Therefore, relaxing this assumption may lead to better scaling. The next section is dedicated to building an algorithm based on this principle, yielding a sub-$\sqrt{d}$ empirical scaling for Gaussian-tailed targets.
\section{A PDMP algorithm breaking the $O(\sqrt{d})$ barrier}
\label{sec:algorithm}

This section puts into practice the idea of relaxing the continuous-time invariance assumption discussed previously and yields an algorithm that breaks the $\mathcal{O}(\sqrt{d})$ complexity barrier. Furthermore, while the algorithm is heavily inspired by HMC, it is naturally locally adaptive in what is the equivalent of the step size. 

Our goal is to draw samples from a target probability measure $\pi(x) \propto \exp(-U(x))$ defined on a state space $\mathcal{X} = \mathbb{R}^d$. We augment the state space with a velocity variable $v \in \mathcal{V} = \mathbb{R}^d$ and define the augmented target measure $\mu(x, v) \propto \exp(-H(x, v))$, where $H(x, v) = U(x) + \frac{1}{2}\|v\|^2$ is the Hamiltonian.

Consider the leapfrog integration steps used in Hamiltonian Monte Carlo (HMC). A leapfrog step of size $\epsilon$ transitions a state $(x_0, v_0)$ to $(x_1, v_1)$ via an intermediate velocity $v_{1/2}$:
\begin{align*}
    v_{1/2} &= v_0 - \frac{\epsilon}{2} \nabla U(x_0), \\
    x_1 &= x_0 + \epsilon v_{1/2}, \\
    v_1 &= v_{1/2} - \frac{\epsilon}{2} \nabla U(x_1).
\end{align*}
If we view this as a continuous trajectory, linking $(x_0, v_0)$ to $(x_0, v_{1/2})$, then to $(x_1, v_{1/2})$, and finally to $(x_1, v_1)$,the Hamiltonian is not conserved along the path. The energy oscillates, but returns to a value close to the initial energy at the end of the step. Figure \ref{fig:leapfrog_oscillation} illustrates this for a 1D target.

\begin{figure}[htbp]
    \centering
    \includegraphics[width=\textwidth]{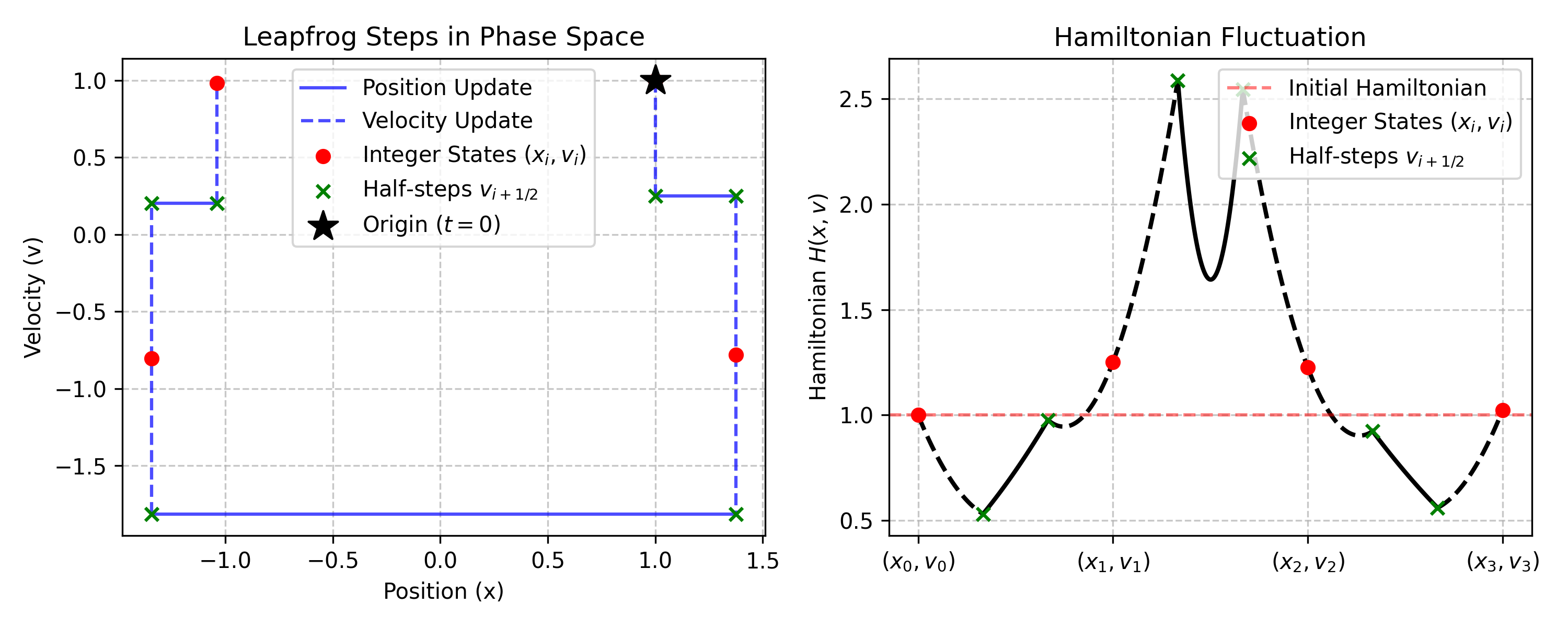}
    \caption{Three leapfrog steps in phase space for a 1D target. The discrete transitions are interpolated linearly (Left), showing significant fluctuations in $H(x, v)$ along the continuous path (Right), despite the energy being approximately conserved between integer states.}
    \label{fig:leapfrog_oscillation}
\end{figure}

Our algorithm uses a similar principle, but transitions between position and velocity updates at random continuous times rather than fixed intervals. While the resulting PDMP does not strictly preserve the target measure $\mu$ along its continuous trajectory, we recover the exact distribution through path re-weighting. Specifically, we construct a probability measure over the trajectory, assigning higher weights when the Hamiltonian is close to its initial value, in such a way that sampling a state directly from it recovers the target distribution. 

This section is organized as follows. Subsections \ref{sec:aux_measure} and \ref{sec:localized_dynamics} construct the underlying dynamics of the process. Subsection \ref{sec:sampling_w_t} introduces the path-resampling measure and details how to sample from it in practice. Finally, Subsection \ref{sec:algo_summary} provides a summary of the algorithm.

\begin{remark}
    We could theoretically build a PDMP sampler directly from the leapfrog integrator to essentially recover standard HMC. By adding a clock variable that increases at rate $1$, we can introduce boundary events every time the clock reaches $\epsilon$. These events simply reset the clock to $0$ and switch the continuous flow between position and velocity updates, as illustrated in Figure \ref{fig:leapfrog_oscillation}. Naturally, this purely deterministic PDMP does not leave the target measure invariant. However, if we simulate the process over an interval $[0, T]$, we can recover the correct distribution by focusing exclusively on the discrete integration times $k\epsilon$. We achieve this by defining a probability measure $w$ on $[0, T]$ proportional to $\sum_{k\epsilon \le T} a_k \delta_{k\epsilon}$, where the weights $a_k$ are the usual HMC Metropolis acceptance ratios, and then sampling a state according to $w$.
\end{remark}

\subsection{The Auxiliary Target Measure}
\label{sec:aux_measure}

Given an initial state $(x_0, v_0)$ with energy $H_0 = H(x_0, v_0)$, we define a localized auxiliary target density concentrated around the initial energy surface:
\begin{equation}
    \label{eq:aux_target}
    \tilde{\mu}_{H_0}(x, v) \propto \exp\left(-\frac{|H(x, v) - H_0|}{2\sigma}\right),
\end{equation}
where $\sigma > 0$ is a scale parameter. Because $\tilde{\mu}_{H_0}$ depends on $H_0$, the specific PDMP simulated changes at every iteration.

For an i.i.d. target in dimension $d$, the variance of the natural Hamiltonian $H(x, v)$ scales as $\mathcal{O}(d)$. Since the variance of $H$ under $\tilde{\mu}_{H_0}$ is proportional to $\sigma^2$, we set $\sigma \propto \sqrt{d}$ to properly match the target variance.

\subsection{The Localized Dynamics}
\label{sec:localized_dynamics}

We extend the phase space to include a flow index $i \in \{1, 2\}$, giving $E = \mathcal{X} \times \mathcal{V} \times \{1, 2\}$. The true target measure is $\mu(x, v, i) = \frac{1}{2}\mu(x, v)$. 

The continuous dynamics are governed by two alternating vector fields:
\begin{itemize}
    \item $\phi(x,v,1) = (v,0,0)$ (position update),
    \item $\phi(x,v,2) = (0, -\nabla U(x),0)$ (velocity update).
\end{itemize}
A deterministic jump map $F$ flips the flow index ($1 \leftrightarrow 2$) without changing the state. To admit $\tilde{\mu}_{H_0}$ as its invariant measure, the process requires the following forward event rates\footnote{The gradient $\nabla \tilde{\mu}_{H_0}$ is undefined exactly at $H(x,v) = H_0$. We define the event rates to be zero at this singularity.}:
\begin{align}
    \lambda^{fwd}(x, v,1) &= \max\left(0, \frac{1}{2\sigma} \text{sgn}\big(H(x, v) - H_0\big) \big(v \cdot \nabla U(x)\big)\right), \\
    \lambda^{fwd}(x, v,2) &= \max\left(0, \frac{1}{2\sigma} \text{sgn}\big(H(x, v) - H_0\big) \big(-v \cdot \nabla U(x)\big)\right).
\end{align}

\begin{proposition}
    The PDMP defined by these vector fields and event rates admits $\tilde{\mu}_{H_0}$ as its invariant measure.
\end{proposition}
\begin{proof}
    See Appendix \ref{app:proofs}.
\end{proof}

This formulation adapts its event rates locally based on the gradient and velocity, leading to a scale invariance property:
\begin{proposition}[Scale Invariance]
    \label{prop:scale-invariance}
    If the target density is scaled as $\tilde{\pi}(x) = \pi(\alpha x)$ for $\alpha > 0$, the generated trajectories $\tilde{Y}(t) = (\tilde{x}(t), \tilde{v}(t))$ match the original trajectories via the rescaling $\tilde{x}(t) = \frac{1}{\alpha} x(\alpha t)$ and $\tilde{v}(t) = v(\alpha t)$.
\end{proposition}

To simulate trajectories backward in time, we negate the vector field ($\phi^{bwd} = -\phi$). The resulting backward switching rates are the forward rates with swapped flow indices: $\lambda^{bwd}(x, v,1) = \lambda^{fwd}(x,v,2)$ and $\lambda^{bwd}(x, v,2) = \lambda^{fwd}(x,v,1)$.

\paragraph{Simulation of the PDMP}
Because the exact continuous-time rates $\lambda^{fwd}$ and $\lambda^{bwd}$ depend continuously on the spatial gradient and the Hamiltonian, simulating the exact event times analytically is generally intractable. To resolve this, we employ the Metropolised PDMP methodology detailed in Section \ref{sec:metropolised_pdmp}. We construct the simulated trajectory using a tractable numerical approximation $\tilde{\lambda}$ of the true rates. A straightforward and computationally efficient choice is a piecewise-constant approximation, where the event rates are evaluated at discrete intervals and held constant between them. Furthermore, the step size of this rate approximation can be made adaptive.

\subsection{Trajectory Generation and Path Weighting}
\label{sec:sampling_w_t}

To generate a continuous trajectory, we start from our current state $(x_0, v_0)$ and sample a uniform time fraction $\alpha \sim \mathcal{U}(0, 1)$. We simulate the PDMP dynamics forward in time for a duration $(1-\alpha)T$ and backward for $\alpha T$. The total trajectory length $T$ is determined dynamically by the NUTS stopping criterion (see Section \ref{sec:nuts_pdmp}).
This procedure generates a trajectory 
\[
    Y = (X_{t-\alpha T},V_{t-\alpha T},i_{t-\alpha T})_{t\in [0,T]}
\]
parameterized on the interval $[0, T]$. By construction, our initial state is located at $Y_l$, where $l = \alpha T$. 

The core theoretical foundation of the algorithm relies on the conditional distribution of this initial starting time $l$ given the generated trajectory $Y$. If we compute the conditional probability density $P(l=t \mid Y)$ and draw a new index $l^*$ from this distribution, the new state $Y_{l^*}$ is guaranteed to have the exact same marginal distribution as the initial state $Y_l$. Since the initial state is assumed to be distributed according to the target $\mu$, the resampled state $Y_{l^*}$ is also an exact sample from $\mu$.

Therefore, our true selection weight $w_t$ must be exactly this conditional density: $w_t \propto P(l=t \mid Y)$. Following the Metropolized PDMP framework (Section \ref{sec:metropolised_pdmp}), evaluating this joint density yields:
\begin{equation}
    \label{eq:wt_joint}
    w_t \propto \mu(Y_t) \tilde{p}_{W_{[t,T]}}(Y_{[t,T]}) \tilde{p}^r_{W_{[0,t]}}\big(R(Y_{[0,t]})\big) \nu(t),
\end{equation}
where, $\tilde{p}_{W}$ and $\tilde{p}^r_{W}$ are the path probabilities of generating the forward and backward segments under the \textit{approximated} numerical event rates used during the simulation, $R$ is the time-reversal operator, and $\nu(t)$ is the volume factor induced by the stopping criterion.

To efficiently sample the new index $l^*$, we use a Metropolis-Hastings (MH) step over the trajectory $[0,T]$. This is possible because evaluating Equation \ref{eq:wt_joint} can be done exactly. However, sampling such a probability distribution using a naive proposal would incure a very low acceptance rate, and therefore a high number of evaluation of Equation \ref{eq:wt_joint} to get an accepted sample. This would be prohibitively expensive as evaluating the true weight $w_t$ point-wise for candidate times requires rebuilding the entire numerical rate approximation anchored from the shifted starting point $Y_t$. Therefore we require a computationally cheap proposal density $q_Y(t)$ that closely approximates the true $w_t$.

To build $q_Y(t)$, we derive an analytical surrogate for the weight, denoted $w_t^{exact}$. This surrogate is computed under the idealized assumption that the trajectory was generated using the \textit{exact} continuous-time rates rather than the numerical approximations. Thanks to our targeted auxiliary measure $\tilde{\mu}_{H_0}$, this idealized conditional density simplifies into an explicit formula.

\begin{proposition}
    \label{prop:w_t}
    For a trajectory $Y$ of duration $T$, if the localized PDMP is simulated using the exact continuous-time rates, the idealized selection weight $w_t^{exact}$ for a time $t \in [0,T]$ is analytically given by:
    \begin{equation}
        \label{eq:w_t_exact}
        w_t^{exact} \propto \mu(Y_t) \nu(t) \exp\left( - \frac{|H_{end} - H(Y_t)| + |H_{start} - H(Y_t)|}{4\sigma} \right) \mathbb{I}\big(H_{min} < H(Y_t) < H_{max}\big),
    \end{equation}
    where $H_{start} = H(Y_0)$ and $H_{end} = H(Y_T)$ are the Hamiltonian values at the backward and forward extremities of $Y$. 
    
    The bounds $H_{min}$ and $H_{max}$ are the empirical local extrema of the energy along the discrete jump path:
    \begin{align}
        H_{max} &= \min \left( \min_{k \in K_{fwd}: \dot{H}(\tau_k^-) > 0} H(\tau_k), \min_{k \in K_{bwd}: \dot{H}(\tau_k^+) < 0} H(\tau_k) \right), \\
        H_{min} &= \max \left( \max_{k \in K_{fwd}: \dot{H}(\tau_k^-) < 0} H(\tau_k), \max_{k \in K_{bwd}: \dot{H}(\tau_k^+) > 0} H(\tau_k) \right),
    \end{align}
    where $K_{fwd}$ and $K_{bwd}$ are the sets of forward and backward event indices, and $\dot{H}(\tau_k^\pm)$ denotes the time derivative of the Hamiltonian immediately before or after the jump.
\end{proposition}

\begin{figure}[htbp]
    \centering
    \includegraphics[width=1.1\textwidth]{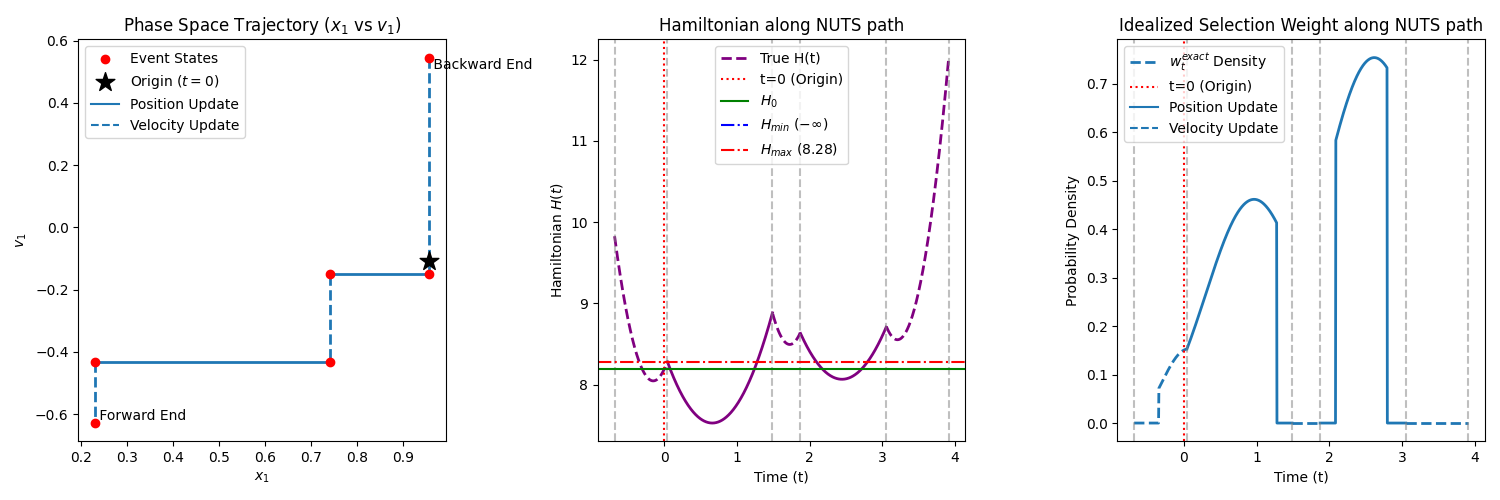}
    \caption{Example trajectory for a 1D target. The idealized weight $w_t^{exact}$ is highest when the Hamiltonian is close to $H_0$. Density increases from left to right due to $\nu(t)$.}
    \label{fig:exact_w_t}
\end{figure}

Figure \ref{fig:exact_w_t} illustrates the computation of $w_t^{exact}$. Because evaluating $w_t^{exact}$ continuously still depends on the exact Hamiltonian, we form the final proposal density $q_Y(t)$ by building a piecewise-quadratic interpolation of $H$ between the discrete event times and substituting this surrogate directly into Equation \ref{eq:w_t_exact}. 
Furthermore, to ensure robust MH exploration, the final proposal $q_Y(t)$ is a mixture comprising $90\%$ of this targeted quadratic approximation and $10\%$ of the baseline density $\nu$  over $[0,T]$ (see Figure \ref{fig:proposal_w_t}).

\begin{figure}[htbp]
    \centering
    \includegraphics[width=1.1\textwidth]{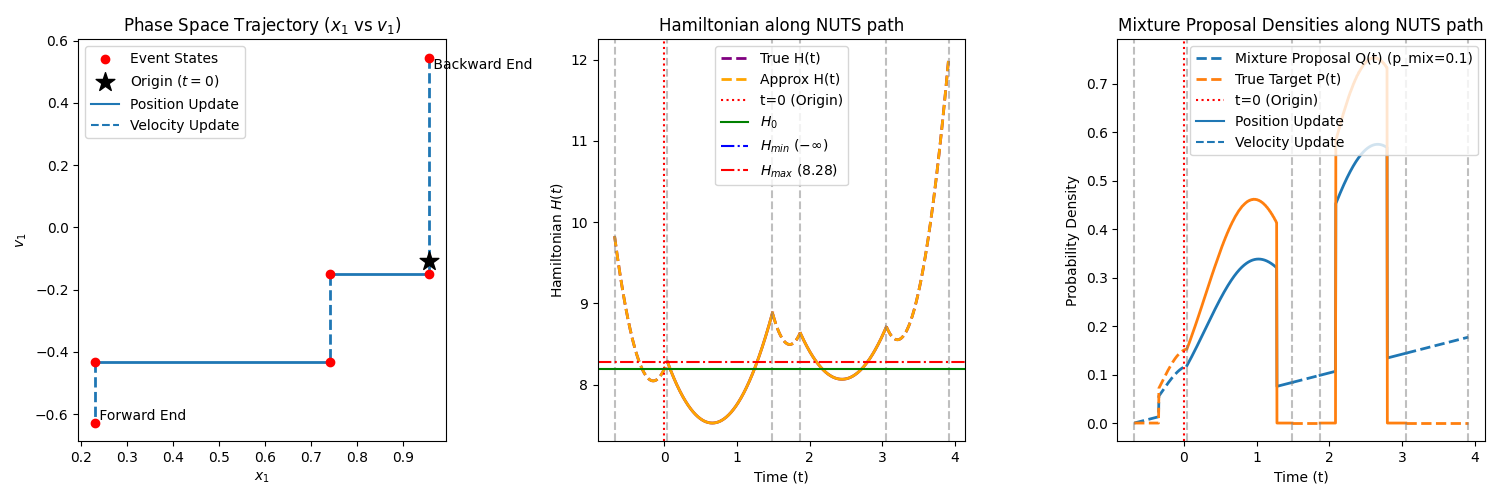}
    \caption{The idealized weight $w_t^{exact}$ (orange) and the corresponding MH proposal distribution $q_Y$ (blue) for a 1D target.}
    \label{fig:proposal_w_t}
\end{figure}

\subsection{Algorithm Summary}
\label{sec:algo_summary}

The full sampling procedure for a single iteration is detailed in Algorithm \ref{alg:pdmp_iteration}.

\begin{algorithm}[htbp]
\caption{Single Iteration of the Metropolized Oscillating PDMP}
\label{alg:pdmp_iteration}
\textbf{Input:} Current state $(x_0, v_0)$ \\
\textbf{Output:} Next state $(x_{next}, v_{next})$
\begin{enumerate}
    \item \textbf{Initialize:} Compute $H_0 = H(x_0, v_0)$ and define $\tilde{\mu}_{H_0}$ (scaling $\sigma$ with $\sqrt{d/2}$).
    \item \textbf{Sample Time Fraction:} Draw $\alpha \sim \mathcal{U}(0, 1)$.
    \item \textbf{Simulate Trajectory:} Simulate the PDMP with approximated event rates forward for $(1-\alpha)T$ and backward for $\alpha T$. The length $T$ is set by the NUTS stopping criterion. Let $Y$ be the generated trajectory on $[0, T]$.
    \item \textbf{Build Proposal:} Compute the idealized weights $w_t^{exact}$ along $Y$ using Proposition \ref{prop:w_t} to construct the mixture proposal density $q_Y(t)$.
    \item \textbf{Resample State:} Use Metropolis-Hastings with proposal $q_Y(t)$ to sample a new time index $t^* \in [0, T]$, targeting the true selection weight $w_t$.
    \item \textbf{Return:} Output $Y_{t^*}$ as the next state.
\end{enumerate}
\end{algorithm}

\section{Empirical results}
\label{sec:empirical_results}

The primary objective of these empirical experiments is to analyze the scaling behavior of the proposed sampler with respect to the dimension $d$ of the target distribution, and show that it breaks the $\Omega(\sqrt{d})$ scaling limite of traditional PDMP samplers for i.i.d. targets shown in section \ref{sec:complexity_bound}. Specifically, we investigate how the computational effort, measured both by the number of simulated events and the total number of gradient evaluations required to obtain one effective sample (i.e., per ESS), scales as $d$ increases. We restrict our analysis to product measures corresponding to i.i.d. targets.

\paragraph{Targets}
Since our algorithm is inspired by HMC, we expect it to exhibit similar behavior. For HMC, Gaussian-tailed targets represent the optimal setting, while lighter or heavier tails can be highly problematic. Therefore, we do not expect the algorithm to break the $\mathcal{O}(d^{1/2})$ barrier for non-Gaussian targets, but we include them to explicitly test the algorithm's robustness. Most notably, we expect our method to perform reliably on light-tailed distributions where standard HMC would simply fail because it allows for adaptive step sizes. To this end, we chose five targets, categorizing them into three with Gaussian-like tails and two with non-Gaussian tails (specifically, one heavier-tailed and one lighter-tailed target).

\subparagraph{Gaussian-Tailed Targets}
These targets serve as our well-behaved benchmarks, dominated by Gaussian tails where standard HMC performs optimally.
\begin{enumerate}
    \item \textbf{Gaussian}: A standard isotropic Gaussian distribution. The target density is given by
    \begin{equation*}
        \pi(x) \propto \exp\left(-\frac{1}{2} \sum_{i=1}^d x_i^2\right).
    \end{equation*}
    
    \item \textbf{Bayesian Logistic Regression (BLR)}: A target representing a Bayesian logistic regression posterior with independent standard normal priors, simplified by assuming all observed labels are equal to 1. The target density is
    \begin{equation*}
        \pi(x) \propto \prod_{i=1}^d \frac{1}{1 + \exp(-x_i)} \exp\left(-\frac{1}{2} x_i^2\right).
    \end{equation*}

    \item \textbf{Gaussian Mixture}: A multimodal target formed by an equal-weighted mixture of two $d$-dimensional standard isotropic Gaussian distributions centered at $\mu_1 = (-2, 0, \dots, 0)$ and $\mu_2 = (2, 0, \dots, 0)$. The target density is
    \begin{equation*}
        \pi(x) \propto \left( \exp\left(-\frac{1}{2} (x_1+2)^2\right) + \exp\left(-\frac{1}{2} (x_1-2)^2\right) \right) \prod_{i=2}^d \exp\left(-\frac{1}{2} x_i^2\right).
    \end{equation*}
\end{enumerate}

\subparagraph{Non-Gaussian-Tailed Targets}
These targets present specific exploration challenges due to their tail decay rates, designed to test the algorithm's stability where fixed-step-size HMC typically struggles.
\begin{enumerate}
    \item \textbf{Logistic (Heavier tails)}: A distribution with heavier tails than the Gaussian, constructed as a product of independent standard logistic distributions. This heavier tail behavior makes exploration challenging for standard HMC. The target density is
    \begin{equation*}
        \pi(x) = \prod_{i=1}^d \frac{e^{-x_i}}{(1 + e^{-x_i})^2}.
    \end{equation*}
    
    \item \textbf{Quartic / Sub-Gaussian (Lighter tails)}: A light-tailed target with a potential proportional to $-0.25 x^4 - 0.5 x^2$. HMC without step-size adaptivity routinely fails here, as the rapidly increasing gradient easily leads to numerical instabilities. The target density is
    \begin{equation*}
        \pi(x) \propto \exp\left(-\sum_{i=1}^d \left(\frac{1}{4} x_i^4 + \frac{1}{2} x_i^2\right)\right).
    \end{equation*}
\end{enumerate}

\paragraph{Empirical Setup}
For each target, we sample in dimensions $d \in \{2, 4, 8, 16, 32, 64, 128\}$. The scale parameter is dynamically set as $\sigma = \sqrt{d/2}$. The numerical approximation of the event rates employs our adaptive approximation method. We draw 10,000 samples for each configuration and quantify the number of independent samples generated by evaluating the Effective Sample Size (ESS) on the first coordinate $x_1$. The ESS computations are performed using the ArviZ library \cite{Martin2026}.

\begin{table}[htbp]
    \centering
    \begin{tabular}{lccl}
        \hline
        \textbf{Target} & \textbf{Number of Events} & \textbf{Number of Gradient Evaluations} & \textbf{Tail Behavior} \\
        \hline
        Gaussian & $\mathcal{O}(d^{0.22})$ & $\mathcal{O}(d^{0.26})$ & \multirow{3}{*}{Gaussian} \\
        BLR & $\mathcal{O}(d^{0.24})$ & $\mathcal{O}(d^{0.29})$ & \\
        Gaussian Mixture & $\mathcal{O}(d^{0.29})$ & $\mathcal{O}(d^{0.31})$ & \\
        \hline
        Logistic & $\mathcal{O}(d^{0.43})$ & $\mathcal{O}(d^{0.48})$ & Heavier \\
        Quartic & $\mathcal{O}(d^{0.61})$ & $\mathcal{O}(d^{0.64})$ & Lighter \\
        \hline
    \end{tabular}
    \caption{Empirical scaling of the Metropolized PDMP sampler with respect to dimension $d$.}
    \label{tab:empirical_scaling}
\end{table}

\paragraph{Results}
The scaling results are summarized in Table \ref{tab:empirical_scaling}. We present the detailed scaling analyses grouped by tail behavior: Figure \ref{fig:scaling_gaussian_tails} illustrates the performance on targets with Gaussian tails, while Figure \ref{fig:scaling_nongaussian_tails} details the results for targets with non-Gaussian tails. The empirical scaling for the number of events for Gaussian-tailed targets lies between $\mathcal{O}(d^{0.22})$ and $\mathcal{O}(d^{0.29})$, close to the theoretical $\mathcal{O}(d^{0.25})$ scaling of standard HMC, and overcoming the $\mathcal{O}(d^{0.5})$ limit established for traditional PDMP algorithms.

The number of gradient evaluations scales slightly less favorably, ranging between $\mathcal{O}(d^{0.26})$ and $\mathcal{O}(d^{0.31})$. This marginal degradation is not unexpected, as the numerical approximation of the event rates can become more challenging in higher dimensions. However, it may also stem from implementation overhead or limitations in the current numerical approximation scheme, suggesting that employing higher-order methods could further optimize performance. Overall, the excess complexity is well-contained, adding only $0.02$ to $0.05$ to the scaling exponent. Finally, the algorithm demonstrates robustness, maintaining reasonable scaling behavior when for both the light-tailed Quartic distribution and the heavier-tailed Logistic distribution.

\begin{figure}[htbp]
    \centering
    \includegraphics[width=\textwidth]{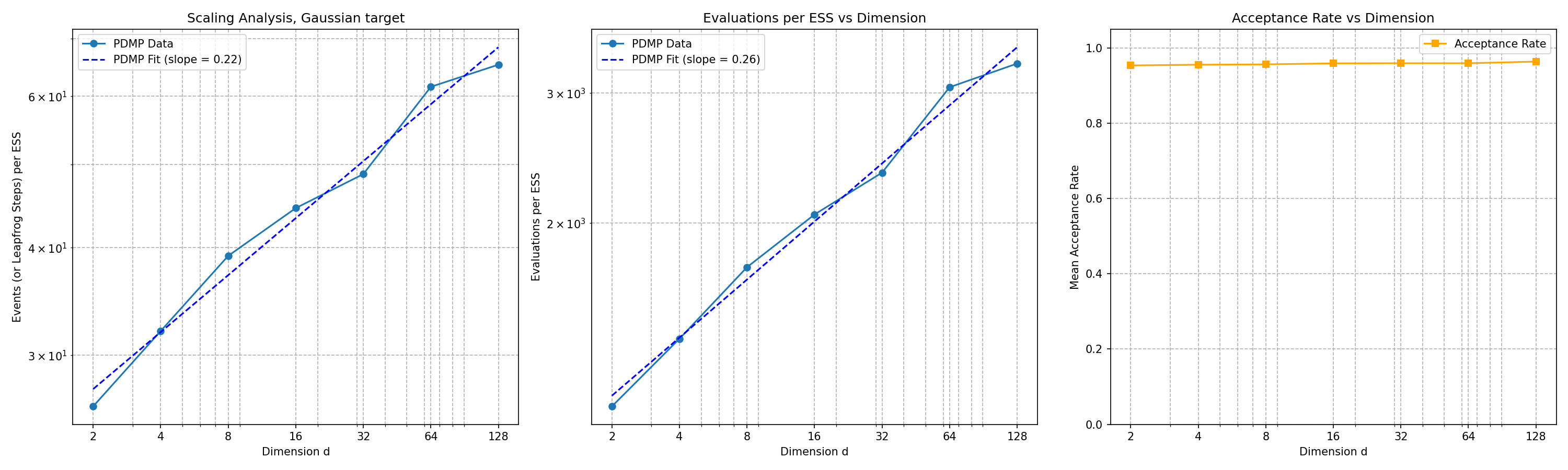} \\[1em]
    \includegraphics[width=\textwidth]{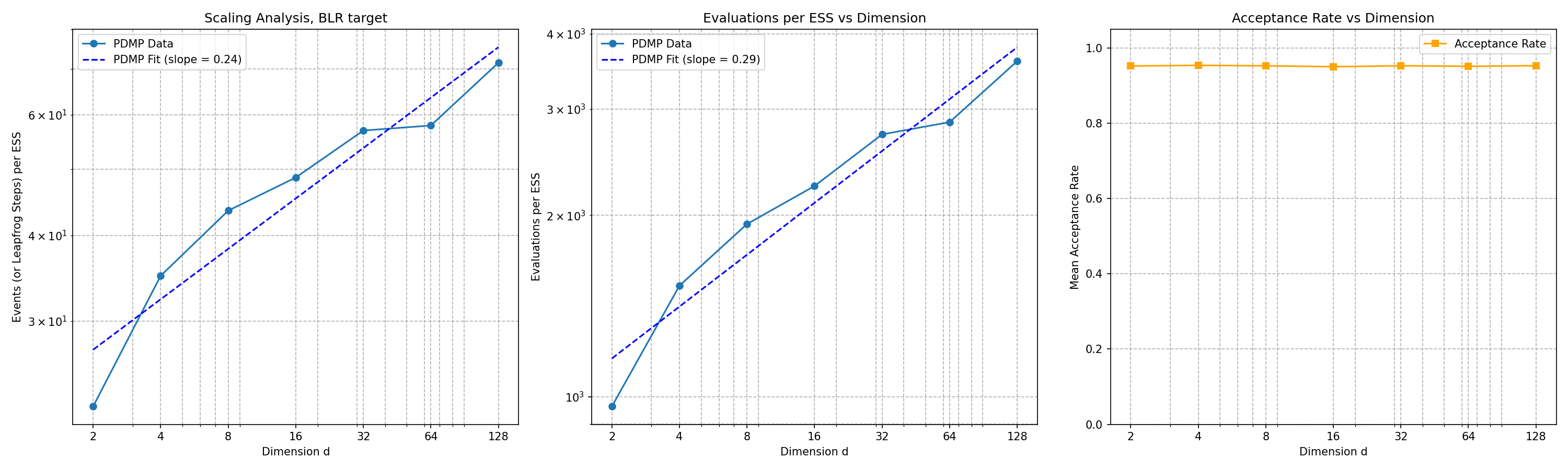} \\[1em]
    \includegraphics[width=\textwidth]{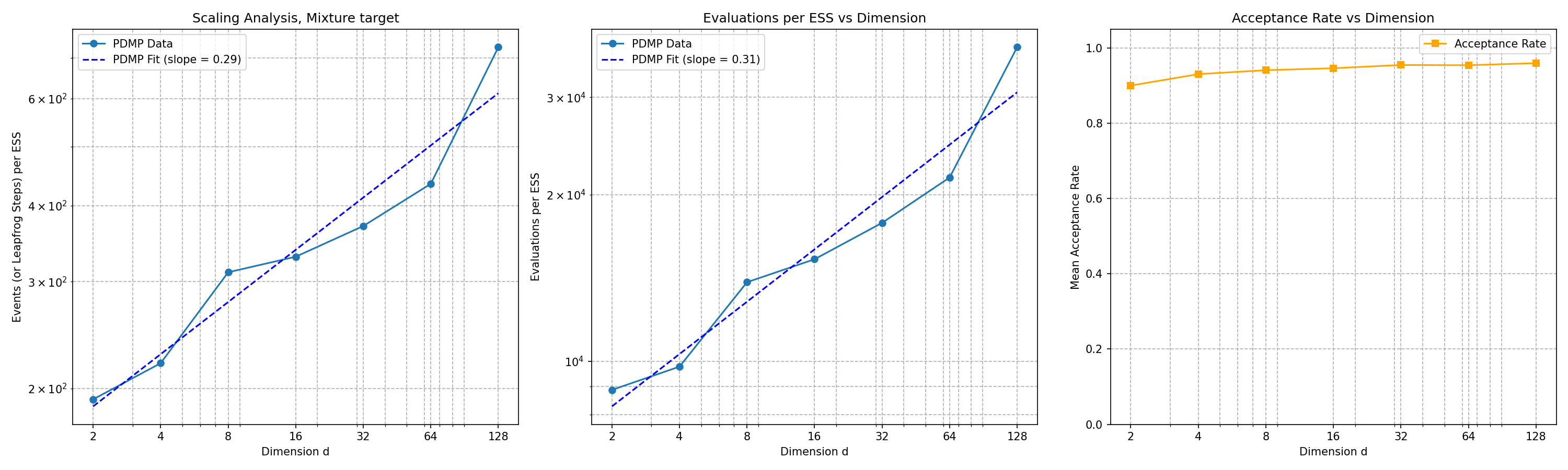}
    \caption{Scaling analysis for baseline targets with \textbf{Gaussian tails}. Top: Standard Gaussian; Middle: Bayesian Logistic Regression (BLR); Bottom: Gaussian Mixture. Each sub-figure displays Events per ESS, Number of Gradient Evaluations per ESS, and Mean Acceptance Rate.}
    \label{fig:scaling_gaussian_tails}
\end{figure}

\begin{figure}[htbp]
    \centering
    \includegraphics[width=\textwidth]{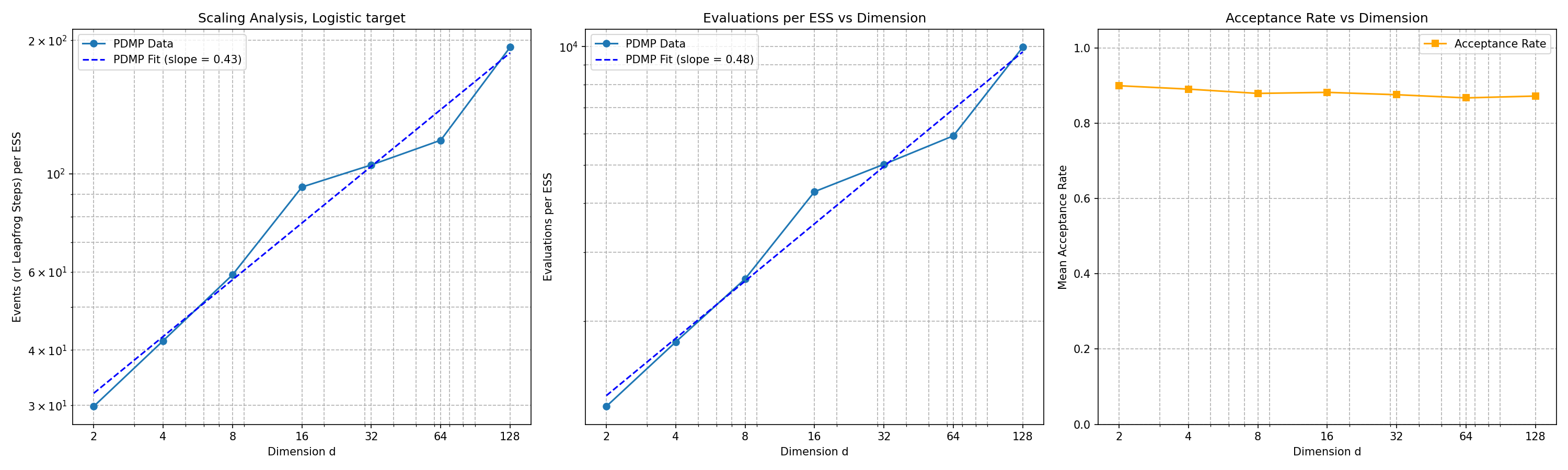} \\[1em]
    \includegraphics[width=\textwidth]{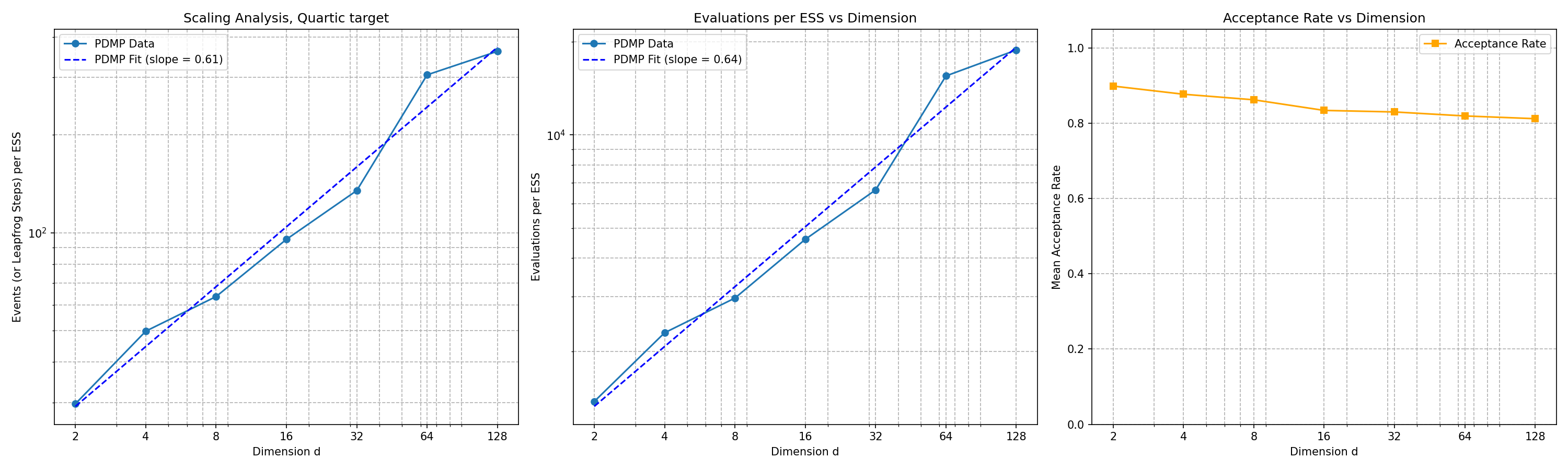}
    \caption{Robustness testing on targets with \textbf{Non-Gaussian tails}. Top: Logistic target (heavier tails); Bottom: Quartic target (lighter tails). Each sub-figure displays Events per ESS, Number of GradientEvaluations per ESS, and Mean Acceptance Rate.}
    \label{fig:scaling_nongaussian_tails}
\end{figure}

\section{Discussion}
\label{sec:discussion}

In this work, we identified and overcame a fundamental scaling limit of Piecewise Deterministic Markov Processes (PDMPs). First, we established an $\Omega(d^{1/2})$ lower bound on the computational complexity of standard PDMPs, showing that the requirement to maintain target invariance at all continuous times strictly limits high-dimensional performance. To explicitly break this barrier, we proposed relaxing the continuous-time invariance requirement. By introducing a path re-weighting strategy, our algorithm successfully bypasses the theoretical lower bound, empirically achieving an $\mathcal{O}(d^\alpha)$ scaling with $\alpha \in [0.2, 0.3]$ for Gaussian-tailed targets.

A key practical advantage of this algorithm is its local adaptivity. In addition to the local adaptation of the trajectory length using the No-U-Turn criterion, the PDMP framework allows the distance between velocity updates to adjust to the local geometry of the target distribution during sampling. While implementing locally adaptive step sizes in standard HMC is challenging due to the need to maintain reversibility, our non-reversible PDMP approach avoids these constraints entirely.

There are several natural directions for future work. First, while we proved the $\mathcal{O}(d^{1/2})$ lower bound under specific assumptions, formally extending this to all local PDMP algorithms remains an open problem. Second, our empirical results suggest a scaling between $\mathcal{O}(d^{1/4})$ and $\mathcal{O}(d^{1/3})$ for Gaussian-tailed targets, but establishing a rigorous theoretical complexity is still needed. The question of optimal tuning also requires further investigation, particularly regarding how to determine the optimal acceptance rate and, consequently, the appropriate numerical tolerance for the approximated event rates.

Finally, there is substantial room for algorithmic exploration within the proposed framework. For instance, the probability density weights, currently based on the absolute difference of Hamiltonians, could be modified to use alternatives like a squared distance penalty. 
%To go further, because our path re-weighting strategy corrects the target distribution, the underlying dynamic does not even require a strict invariant distribution; it only needs to oscillate appropriately around the initial Hamiltonian. 
Additionally, while our algorithm is inspired by Hamiltonian Monte Carlo (HMC) dynamics, future research could explore dynamics different from HMC.
Ultimately, by introducing these new ideas for sampling algorithms, we hope to open the door to the development of new samplers, in particular locally adaptive ones.

\section*{Acknowledgments}

I would like to thank Matthew Sutton, Sam Power and Nicolas Chevallier for their insightful discussions on this work.

\bibliographystyle{plain}
\bibliography{main}

\appendix
\section{Deferred Proofs}
\label{app:proofs}

\subsection{Complexity lower bound}

\subsection{Detailed proof of Proposition \ref{prop:complexity_lower_bound}}
\label{app:proofs_proposition_complexity_bound}

We start by two useful lemmas before the proof of the proposition itself.

\begin{lemma}[Lower Bound on Expected Distance for i.i.d. targets]
    \label{lem:distance_lower}
    Let $X, Y \in \mathbb{R}^d$ be independent random vectors drawn from a product measure $\pi = \pi_1^{\otimes d}$. Assume $\pi_1$ has finite variance and is not a Dirac mass. Then, there exists a constant $C > 0$ such that $\mathbb{E}_\pi[\|X - Y\|_2] \ge C \sqrt{d}$.
\end{lemma}
\begin{proof}
    Applying the standard norm inequality $\|v\|_2 \ge d^{-1/2} \|v\|_1$ to the vector $X - Y$, and taking the expectation, we obtain:
    \[
        \mathbb{E}[\|X - Y\|_2] \ge \frac{1}{\sqrt{d}} \sum_{i=1}^d \mathbb{E}[|X_i - Y_i|].
    \]
    Because $X$ and $Y$ are independently drawn from $\pi_1^{\otimes d}$, the pairs $(X_i, Y_i)$ are i.i.d., which simplifies the sum to:
    \[
        \mathbb{E}[\|X - Y\|_2] \ge \sqrt{d} \, \mathbb{E}[|X_1 - Y_1|].
    \]
    Setting $C = \mathbb{E}[|X_1 - Y_1|]$, we observe that $C > 0$ strictly, as $X_1, Y_1$ are independent and $\pi_1$ is not a Dirac mass.
\end{proof}

\begin{lemma}[Number of events]
    \label{lemma:number_of_events}
    Let $(X_t)_{t \geq 0}$ be a PDMP satisfying assumption \ref{assum:ergodicity} (Ergodicity), and let $N(t)$ be the number of jumps in the time interval $(0, t]$. Assuming that $\mathbb{E}_{\mu}[\lambda] > 0$, then
    \[
    \lim_{t \to \infty} \frac{N(t)}{t} = \mathbb{E}_{\mu}[\lambda] \quad a.s.
    \]
\end{lemma}
\begin{proof}
    Let $\Lambda(t) = \int_0^t \lambda(X_s,Z_s) ds$ be the compensator of the counting process $N(t)$. By proposition 26.6 from \cite{davis1993markov}, the compensated process 
    \[
    M(t) = N(t) - \Lambda(t)
    \]
    is a local martingale. Furthermore, its predictable quadratic variation is exactly its compensator, $\langle M \rangle_t = \Lambda(t)$.
    
    By the ergodicity of $(X_t,Z_t)_{t \geq 0}$ with invariant probability measure $\mu$, the time average of the state-dependent jump intensity converges almost surely to its spatial average:
    \[
    \lim_{t \to \infty} \frac{\langle M \rangle_t}{t} = \lim_{t \to \infty} \frac{\Lambda(t)}{t} = \mathbb{E}_{\mu}[\lambda] \quad a.s.
    \]
    
    Since $\mathbb{E}_{\mu}[\lambda] > 0$, it follows that $\langle M \rangle_t \to \infty$ almost surely as $t \to \infty$. 
    
    Applying the Strong Law of Large Numbers for martingales (see, e.g., Liptser \& Shiryaev, \emph{Theory of Martingales}, Corollary 1.1, p. 144 \cite{liptser1989theory}), we have:
    \[
    \lim_{t \to \infty} \frac{M(t)}{\langle M \rangle_t} = 0 \quad a.s.
    \]
    Substituting $M(t) = N(t) - \Lambda(t)$ and $\langle M \rangle_t = \Lambda(t)$ into the limit yields:
    \[
    \lim_{t \to \infty} \frac{N(t) - \Lambda(t)}{\Lambda(t)} = 0 \implies \lim_{t \to \infty} \frac{N(t)}{\Lambda(t)} = 1 \quad a.s.
    \]
    
    Finally, decomposing the asymptotic jump rate gives:
    \[
    \lim_{t \to \infty} \frac{N(t)}{t} = \left( \lim_{t \to \infty} \frac{N(t)}{\Lambda(t)} \right) \left( \lim_{t \to \infty} \frac{\Lambda(t)}{t} \right) = 1 \cdot \mathbb{E}_{\mu}[\lambda] = \mathbb{E}_{\mu}[\lambda] \quad a.s.
    \]
    This concludes the proof.
\end{proof}

\paragraph{Proof of the proposition}
\begin{proof}
    Consider a single, infinitely long trajectory of the PDMP, $(X_t, Z_t)_{t \geq 0}$, and let $(t_k)_{k \geq 1}$ be the sequence of random times given by Assumption \ref{assum:independent_times}. We first fix the dimension $d$ and evaluate the process up to a macroscopically large stopping time $T = t_K$.
    
    Because the spatial trajectory $t \mapsto X_t$ is absolutely continuous, the total arc length $L(t_K)$ traveled by the process in $\mathbb{R}^d$ up to time $t_K$ is exactly the integral of the speed:
    \[
        L(t_K) = \int_0^{t_K} \|\phi_x(X_s, Z_s)\| ds
    \]
    By the triangle inequality, the path length traveled between time $t_k$ and $t_{k+1}$ is bounded below by the Euclidean distance between the two endpoints. Summing these bounds over all intervals up to $t_K$ gives:
    \[
        \int_0^{t_K} \|\phi_x(X_s, Z_s)\| ds = \sum_{k=1}^{K-1} \int_{t_k}^{t_{k+1}} \|\phi_x(X_s, Z_s)\| ds \geq \sum_{k=1}^{K-1} \|X_{t_{k+1}} - X_{t_k}\|
    \]
    Dividing both sides by the total elapsed time $t_K$, we obtain:
    \[
        \frac{1}{t_K} \int_0^{t_K} \|\phi_x(X_s, Z_s)\| ds \geq \left( \frac{K}{t_K} \right) \frac{1}{K} \sum_{k=1}^{K-1} \|X_{t_{k+1}} - X_{t_k}\|
    \]
    We now take the limit as $K \to \infty$, which implies $t_K \to \infty$ almost surely. By Assumption \ref{assum:ergodicity}, the left-hand side converges almost surely to $\mathbb{E}_\mu[\|\phi_x\|]$. By Assumption \ref{assum:independent_times}, the empirical average on the right-hand side converges almost surely to $\mathcal{D}_d = \mathbb{E}_{\pi \otimes \pi}[\|X - Y\|]$. Thus, passing to the limit yields:
    \[
        \mathbb{E}_\mu[\|\phi_x\|] \geq \left( \limsup_{K \to \infty} \frac{K}{t_K} \right) \mathcal{D}_d
    \]
    Rearranging this inequality establishes a lower bound on the asymptotic time required per independent sample:
    \[
        \liminf_{K \to \infty} \frac{t_K}{K} \geq \frac{\mathcal{D}_d}{\mathbb{E}_\mu[\|\phi_x\|]}
    \]
    We now bound the sampling complexity $C_d$, which is the limit of the number of jumps $N(t_K)$ per independent sample $K$. By writing $C_d$ as the product of the jump rate and the time per sample, we apply lemma \ref{lemma:number_of_events}:
    \[
        C_d = \lim_{K \to \infty} \frac{N(t_K)}{K} = \lim_{K \to \infty} \left( \frac{N(t_K)}{t_K} \cdot \frac{t_K}{K} \right) \geq \mathbb{E}_\mu[\lambda] \frac{\mathcal{D}_d}{\mathbb{E}_\mu[\|\phi_x\|]}
    \]
    Finally, by Lemma \ref{lem:distance_lower}, we have $\mathcal{D}_d \geq C \sqrt{d}$. Substituting this lower bound into our inequality yields:
    \[
        C_d \geq \frac{C \sqrt{d}}{\mathbb{E}_\mu[\|\phi_x\|]} \mathbb{E}_\mu[\lambda]
    \]
    This concludes the proof.
\end{proof}

\subsubsection{Proof of Lemma \ref{lemma:zero-mean-div}}
\begin{proof}[Proof of Lemma \ref{lemma:zero-mean-div}]
    
    The infinitesimal generator $\mathcal{A}$ of the PDMP acts on a suitable test function $f$ as follows:
    \[
        \mathcal{A}f(x) = \langle \phi(x), \nabla f(x) \rangle + \lambda(x) \int_E (f(y) - f(x)) Q(x, dy).
    \]
    The adjoint evaluated at $\mu$ is given by
    \[
        \mathcal{A}^*\mu(x) = -\text{div}(\phi(x)\mu(x)) + \int_E \lambda(y) \mu(y) Q(y, dx) - \lambda(x)\mu(x).
    \]
    Since $\mu$ is an invariant probability distribution, the stationary equation dictates that the adjoint of the generator applied to $\mu$ must equal zero: $\mathcal{A}^*\mu = 0$. 
        
    We integrate the stationary equation $\mathcal{A}^*\mu(x) = 0$ over the entire state space $E$ with respect to the Lebesgue measure $dx$:
    \[
        -\int_E \text{div}(\phi(x)\mu(x)) dx + \int_E \left( \int_E \lambda(y) \mu(y) Q(y, dx) \right) - \int_E \lambda(x)\mu(x) dx = 0
    \]
    By applying Fubini's theorem to the double integral, we can integrate with respect to $dx$ first. Because $Q(y, \cdot)$ is a probability transition kernel, we know that $\int_E Q(y, dx) = 1$ for any $y$. Thus, the incoming jump term simplifies to:
    \[
        \int_E \lambda(y) \mu(y) \left( \int_E Q(y, dx) \right) dy = \int_E \lambda(y) \mu(y) dy
    \]
    Notice that this exactly mirrors the outgoing jump term $\int_E \lambda(x)\mu(x) dx$ (up to a dummy variable change). Therefore, the total jump terms cancel each other out perfectly:
    \[
        -\int_E \text{div}(\phi(x)\mu(x)) dx + \int_E \lambda(y) \mu(y) dy - \int_E \lambda(x)\mu(x) dx = 0
    \]
    \[
        \int_E \text{div}(\phi(x)\mu(x)) dx = 0
    \]
    Finally, we can rewrite this integral as an expectation with respect to the measure $\mu(dx) = \mu(x) dx$ by multiplying and dividing by $\mu(x)$ inside the integral:
    \[
        \int_E \frac{\text{div}(\phi(x)\mu(x))}{\mu(x)} \mu(x) dx = E_\mu\left[\frac{\text{div}(\phi \mu)}{\mu}\right] = 0
    \]
    This completes the proof.
\end{proof}

\begin{comment}
\subsection{Proof of Proposition \ref{prop:complexity}}
\begin{proof}[Proof of Proposition \ref{prop:complexity}]
    Let $T$ be the expected continuous time required for the process to decorrelate (i.e., to generate a new independent sample). By the ergodic theorem, the time-averaged jump rate along the process trajectory converges to the spatial expectation $E_\mu[\lambda]$. Thus, the expected number of jump events that occur during this decorrelation time $T$ is:
    \[
        C = T \cdot E_\mu[\lambda].
    \]
    
    By definition, the average distance the process travels in the target space $\mathbb{R}^d$ during this time $T$ is $\tau(\mu)$. Similarly invoking the ergodic theorem, the expected speed of the deterministic flow in $\mathbb{R}^d$ is given by $E_\mu[\|\phi_x\|]$. The expected time $T$ can therefore be expressed as the distance divided by the average speed:
    \[
        T = \frac{\tau(\mu)}{E_\mu[\|\phi_x\|]}.
    \]
    
    Substituting this expression for $T$ into the equation for $C$ directly yields the result:
    \[
        C = \frac{\tau(\mu)}{E_\mu[\|\phi_x\|]} E_\mu[\lambda].
    \]
\end{proof}
\end{comment}

\subsection{Algorithm}

\subsubsection{Proof of proposition \ref{prop:w_t}}
\begin{proof}
    The probability density $w_t$ relies on evaluating $P_{H_{Y_t}}(Y)$, which can be decomposed into a continuous boundary integral and a product over discrete event rates.
    
    \paragraph{Continuous Boundary Integral.}
    Using the identity $\max(0, x) = \frac{1}{2}(x + |x|)$, the continuous integral over the path rates splits into a symmetric total variation term and an exact boundary difference. Because the trajectory $Y$ is generated in two opposite directions originating from $y = Y_t$, we integrate the forward process to $H_{end}$ and the backward process to $H_{start}$:
    \begin{align*}
        \int_{0}^{T_{fwd}} \lambda_{H_y}^{fwd}(s) ds + \int_{0}^{T_{bwd}} \lambda_{H_y}^{bwd}(u) du 
        &= \frac{1}{4\sigma} V(Y) + \frac{1}{4\sigma} \int_{H_y}^{H_{end}} \text{sgn}(H - H_y) dH + \frac{1}{4\sigma} \int_{H_y}^{H_{start}} \text{sgn}(H - H_y) dH \\
        &= \frac{1}{4\sigma} V(Y) + \frac{|H_{end} - H_y|}{4\sigma} + \frac{|H_{start} - H_y|}{4\sigma}.
    \end{align*}
    Since the total variation $V(Y)$ is a property of the full path and is invariant to the starting point $y$, the probability contribution from the continuous flows is exactly:
    \begin{equation}
        \exp\left(-\int_0^{T_{fwd}} \lambda_{H_y}^{fwd}(s) ds - \int_0^{T_{bwd}} \lambda_{H_y}^{bwd}(u) du\right) \propto \exp\left( - \frac{|H_{end} - H_y| + |H_{start} - H_y|}{4\sigma} \right).
    \end{equation}

    \paragraph{Discrete Event Rates and Energy Bounds.}
    The product of the jump rates at the event times $\tau_k$ restricts the domain of valid starting points. Let $\dot{H}(\tau_k)$ be the physical time derivative of the Hamiltonian at event time $\tau_k$. 
    
    For a forward-generated event ($k \in K_{fwd}$), the rate evaluates to $\lambda_{H_y}^{fwd}(\tau_k) = \max\left(0, \frac{1}{2\sigma}\text{sgn}(H(\tau_k) - H_y)\dot{H}(\tau_k)\right)$. This rate is strictly positive if and only if $(H(\tau_k) - H_y)$ and $\dot{H}(\tau_k)$ share the same sign.
    \begin{itemize}
        \item If $\dot{H}(\tau_k) > 0$, we require $H(\tau_k) - H_y > 0$, hence $H_y < H(\tau_k)$.
        \item If $\dot{H}(\tau_k) < 0$, we require $H(\tau_k) - H_y < 0$, hence $H_y > H(\tau_k)$.
    \end{itemize}

    For a backward-generated event ($k \in K_{bwd}$), the process is simulated backwards in time, effectively reversing the derivative to $-\dot{H}(\tau_k)$. The rate evaluates to $\lambda_{H_y}^{bwd}(\tau_k) = \max\left(0, \frac{1}{2\sigma}\text{sgn}(H(\tau_k) - H_y)(-\dot{H}(\tau_k))\right)$. This is strictly positive if and only if $(H(\tau_k) - H_y)$ and $-\dot{H}(\tau_k)$ share the same sign.
    \begin{itemize}
        \item If $\dot{H}(\tau_k) > 0$, we require $H(\tau_k) - H_y < 0$, hence $H_y > H(\tau_k)$.
        \item If $\dot{H}(\tau_k) < 0$, we require $H(\tau_k) - H_y > 0$, hence $H_y < H(\tau_k)$.
    \end{itemize}

    Because the magnitude of the jump rate evaluates strictly to $\frac{1}{2\sigma}|\dot{H}(\tau_k)|$ regardless of $H_y$ whenever it is non-zero, the product over all jumps acts purely as a binary indicator function. Combining the constraints above across all events yields the global upper bound $H_{max}$ and lower bound $H_{min}$ on $H_y$:
    \begin{equation}
        \left(\prod_{k \in K_{fwd}} \lambda_{H_y}^{fwd}(\tau_k)\right) \left(\prod_{k \in K_{bwd}} \lambda_{H_y}^{bwd}(\tau_k)\right) \propto \mathbb{I}\big(H_{min} < H(y) < H_{max}\big).
    \end{equation}
    Multiplying this indicator by the continuous thermodynamic factor and the prior $\mu(Y_t)$ recovers the exact expression for $w_t$, concluding the proof.
\end{proof}

\subsection{Proof of Proposition \ref{prop:scale-invariance}}
\begin{proof}
    Consider the rescaled variables $\tilde{x}(t) = \frac{1}{\alpha} x(\alpha t)$ and $\tilde{v}(t) = v(\alpha t)$, and the rescaled potential $\tilde{U}(\tilde{x}) = U(\alpha \tilde{x})$. Under Flow 1, $\frac{d\tilde{x}}{dt} = \frac{1}{\alpha} \alpha \dot{x}(\alpha t) = \tilde{v}(t)$, matching the required dynamics. Under Flow 2, $\frac{d\tilde{v}}{dt} = \alpha \dot{v}(\alpha t) = -\alpha \nabla U(x(\alpha t)) = -\nabla \tilde{U}(\tilde{x}(t))$. 
    
    The Hamiltonian evaluated on the rescaled trajectory satisfies $\tilde{H}(\tilde{x}(t), \tilde{v}(t)) = U(\alpha \tilde{x}(t)) + \frac{1}{2} \|\tilde{v}(t)\|^2 = H(x(\alpha t), v(\alpha t))$. The time derivative of the Hamiltonian is therefore $\frac{d\tilde{H}}{dt}(t) = \alpha \dot{H}(\alpha t)$. 
    
    Consequently, the jump rate evaluated along the rescaled trajectory is $\tilde{\lambda}(t) = \alpha \lambda(\alpha t)$. The accumulated rate between any two scaled times $t_0$ and $t_1$ evaluates to $\int_{t_0}^{t_1} \alpha \lambda(\alpha s) ds = \int_{\alpha t_0}^{\alpha t_1} \lambda(u) du$, which is identical to the original accumulated rate. Thus, the sequence of jump events and exponential triggers precisely mirrors the original process. The algorithm simply runs $\alpha$ times faster in physical time while exploring the $\frac{1}{\alpha}$-scaled space, achieving exact scale invariance.
\end{proof}

\end{document}